\documentclass[runningheads]{llncs}

\usepackage{amssymb}
\usepackage{xspace}
\usepackage[english]{babel}
\usepackage{amsmath}
\usepackage{latexsym}
\usepackage{epsfig}
\usepackage{subfigure}
\usepackage{times}
\usepackage{enumerate}
\usepackage{tabularx}
\usepackage{caption}
\usepackage[plainpages=false]{hyperref}
\usepackage{multirow}
\usepackage{hyphenat}
\usepackage{amsfonts}
\usepackage{dsfont}
\usepackage{compress}
\usepackage{paralist}
\usepackage{multirow}
\usepackage[table]{xcolor}

\definecolor{blue}{rgb}{0.274,0.392,0.666}
\definecolor{red}{rgb}{0.627,0.117,0.156}
\definecolor{green}{rgb}{0,0.588,0.509}

\newcommand{\red}[1]{{\color{red}{#1\xspace}}}
\newcommand{\blue}[1]{{\color{blue}{#1\xspace}}}
\newcommand{\green}[1]{{\color{green}{#1\xspace}}}
\newcommand{\NP}{$\mathcal{NP}$\xspace}

\newcommand{\NPC}{\mbox{\NP-complete}\xspace}

\newcommand{\NPCN}{\mbox{\NP-completeness}\xspace}

\newcommand{\NPHN}{\mbox{\NP-hardness}\xspace}

\newcommand{\sefe}{SEFE\xspace}

\newcommand{\Er}{\textcolor{red}{$E_1$}\xspace}
\newcommand{\Eb}{\textcolor{blue}{$E_2$}\xspace}

\newcommand{\Gintm}{G_{\cap}}
\newcommand{\Gint}{$G_{\cap}$\xspace}
\newcommand{\Gun}{$G_{\cup}$\xspace}
\newcommand{\Gr}{\textcolor{red}{$G_1$}\xspace}
\newcommand{\GammaR}{\textcolor{red}{$\Gamma_1$}\xspace}

\newcommand{\Gb}{\textcolor{blue}{$G_2$}\xspace}
\newcommand{\GammaB}{\textcolor{blue}{$\Gamma_2$}\xspace}
\newcommand{\GbP}[1]{\textcolor{blue}{$G^{#1}_{2}$}\xspace}

\newcommand{\Gg}{\textcolor{green}{$G_3$}\xspace}
\newcommand{\GammaG}{\textcolor{green}{$\Gamma_3$}\xspace}
\newcommand{\sunsefeinstance}{$\langle\textcolor{red}{G_1},\textcolor{blue}
{G_2},\textcolor{green}{G_3}\rangle$\xspace}
\newcommand{\sunsefeinstancep}[1]{$\langle\textcolor{red}{G_1^{#1}},\textcolor{blue}
{G_2^{#1}},\textcolor{green}{G_3^{#1}}\rangle$\xspace}
\newcommand{\sunsefesolution}{$\langle\textcolor{red}{\Gamma_1},\textcolor{blue}
{\Gamma_2},\textcolor{green}{\Gamma_3}\rangle$\xspace}

\newcommand{\ptckpbeinstance}{$\langle T, E_1,\dots, E_k \rangle$\xspace}

\newcommand{\ptcTWOpbeinstance}{$\langle T, \textcolor{red}{E_1},\textcolor{blue}{E_2} \rangle$\xspace}

\newcommand{\ptcTWOpbeinstancep}[1]{$\langle T^{#1}, \textcolor{red}{E_1^{#1}},\textcolor{blue}{E_2^{#1}} \rangle$\xspace}
\newcommand{\Erp}[1]{\textcolor{red}{$E_1^{#1}$}\xspace}
\newcommand{\Ebp}[1]{\textcolor{blue}{$E_2^{#1}$}\xspace}
\newcommand{\stplong}{{\scshape Planar Steiner Tree}\xspace}
\newcommand{\stp}{{\scshape PST}\xspace}
\newcommand{\upstp}{{\scshape Uniform Triconnected PST}\xspace}
\newcommand{\xorsatp}{{\scshape MAX $2$-XorSat}\xspace}

\newcommand{\upstshort}{{\scshape UTPST}\xspace}
\newcommand{\betp}{{\scshape Betweenness}\xspace}
\newcommand{\sefep}{{\scshape SEFE}\xspace}
\newcommand{\maxsefep}{{\scshape Max SEFE}\xspace}
\newcommand{\sunsefep}{{\scshape Sunflower SEFE}\xspace}

\newcommand{\ptckpbep}{{\scshape Partitioned T-Coherent $k$-Page Book embedding}\xspace}

\newcommand{\ptcTWOpbep}{{\scshape Partitioned T-Coherent $2$-Page Book embedding}\xspace}

\newcommand{\pkpbep}{{\scshape Partitioned $k$-Page Book embedding}\xspace}

\newcommand{\pkpbepshort}{{\scshape PBE-$k$}\xspace}
\newcommand{\pTWOpbepshort}{{\scshape PBE-$2$}\xspace}
\newcommand{\pTHREEpbepshort}{{\scshape PBE-$3$}\xspace}

\newcommand{\ptckpbepshort}{{\scshape PTBE-$k$}\xspace}

\newcommand{\ptcTHREEpbepshort}{{\scshape PTBE-$3$}\xspace}
\newcommand{\ptcTWOpbepshort}{{\scshape PTBE-$2$}\xspace}

\newcommand{\sefesolution}{$\langle\textcolor{red}{\Gamma_1},\textcolor{blue}
{ \Gamma_2 } \rangle$\xspace}
\newcommand{\maxsefeinstance}[1]{$\langle\textcolor{red}{G_1},\textcolor{
blue} {G_2},#1\rangle$\xspace}

\newcommand{\remove}[1]{}

\newenvironment{proofsketch}
{{\bf Proof sketch:}}{\hspace*{\fill}$\Box$\par\vspace{2mm}}

\newtheorem{cl}{Claim}

\makeatletter
\@ifclassloaded{elsarticle}{
\newtheorem{theorem}{Theorem}
\newtheorem{lemma}{Lemma}
\newtheorem{corollary}{Corollary}
\newenvironment{proof}{{\em Proof.}}
}{}
\makeatother

\newif\ifllncsvar

\makeatletter \@ifclassloaded{llncs}{ \llncsvartrue

\authorrunning{Angelini {\em et al.}}
\titlerunning{Advancements on SEFE and Partitioned Book Embedding Problems}

   \author{Patrizio Angelini$^*$, Giordano {Da Lozzo}$^*$, Daniel Neuwirth$^{**}$}

  \title{Advancements on SEFE and\\ Partitioned Book Embedding Problems}

  \institute{
    $^*$~Department of Engineering, Roma Tre University, Italy\\
    \email{\{angelini,dalozzo\}@dia.uniroma3.it} \\
    $^{**}$~Universit\"at Passau, Germany \\
    \email{daniel.neuwirth@uni-passau.de}
  }
  \bibliographystyle{splncs03}}{
  \title{SEFE $=$ C-Planarity?}  \usepackage{ecrc} \volume{00}
  \firstpage{1}
  \runauth{Angelini {\em et al.}}
  \usepackage{amssymb} \usepackage[figuresright]{rotating}

\newenvironment{proof}
{{\bf Proof:}}{\hspace*{\fill}$\Box$\par\vspace{2mm}}
 } \makeatother

\let\nu\undefined \DeclareMathOperator{\nu}{\theta}

\begin{document}

\ifllncsvar
\maketitle
\else
\begin{frontmatter}
  \fi

\begin{abstract}
  In this work we investigate the complexity of some problems related to the
  \emph{Simultaneous Embedding with Fixed Edges} (\sefe) of $k$ planar graphs and the \pkpbep (\pkpbepshort) problems, which are known to be equivalent under certain conditions.

  While the computational complexity of \sefe for $k=2$ is still a
  central open question in Graph Drawing, the problem is \NPC for $k
  \geq 3$~[Gassner \emph{et al.}, WG '06], even if the intersection
  graph is the same for each pair of graphs (\emph{sunflower
    intersection})~[Schaefer, JGAA (2013)].

  We improve on these results by proving that \sefe with $k \geq 3$
  and sunflower intersection is \NPC even when the intersection graph
  is a tree and all the input graphs are biconnected.
  Also, we prove \NPCN for $k \geq 3$ of problem \pkpbepshort and of problem \ptckpbep (\ptckpbepshort) - that is the generalization of \pkpbepshort in which the ordering of the vertices on the spine is constrained by a tree $T$ - even when two input graphs are biconnected.
  Further, we provide a linear-time algorithm for \ptckpbepshort when $k-1$ pages are assigned a connected graph.  
  Finally, we prove that the problem of maximizing the number of edges
  that are drawn the same in a \sefe of two graphs is \NPC in several restricted settings (\emph{optimization version of \sefe}, Open Problem $9$, Chapter $11$ of the
  Handbook of Graph Drawing and Visualization).
\end{abstract}

\ifllncsvar \else

\author{Patrizio Angelini}
\author{Giordano {Da Lozzo}}
\author{Daniel Neuwirth}

\address{{\small Department of Engineering, Roma Tre University, Italy}\\
  {\{angelini,dalozzo\}@dia.uniroma3.it}\\
    Universit\"at Passau, Germany \\
    daniel.neuwirth@uni-passau.de}


\end{frontmatter}
\fi

\section{Introduction}\label{se:introduction}

Let $G_1,\dots,G_k$ be $k$ graphs on the same set $V$ of vertices. A
\emph{simultaneous embedding with fixed edges} (\sefe) of
$G_1,\dots,G_k$ consists of $k$ planar drawings
$\Gamma_1,\dots,\Gamma_k$ of $G_1,\dots,G_k$, respectively, such that
each vertex $v\in V$ is mapped to the same point in every drawing
$\Gamma_i$ and each edge that is common to more than one graph is
represented by the same simple curve in the drawings of all such
graphs. The \emph{\sefe problem} is the problem of testing whether $k$
input graphs $G_1,\dots,G_k$ admit a \sefe~\cite{ek-sepgfb-05}.

The possibility of drawing together a set of graphs gives the
opportunity to represent at the same time a set of different binary
relationships among the same objects, hence making this topic an
fundamental tool in Information Visualization~\cite{ekln-sgdlav-05}.
Motivated by such applications and by their theoretical appealing,
simultaneous graph embeddings received wide research attention in the
last few years. For an up-to-date survey, see~\cite{bkr-sepg-12}.

Recently, a new major milestone to assert the importance of \sefe has
been provided by Schaefer~\cite{s-ttphtpv-13}, who discussed its
relationships with some other famous problems in Graph Drawing,
proving that \sefe generalizes several of them. In particular, he
showed a polynomial-time reduction to \sefe with $k=2$ from the
\emph{clustered planarity testing}
problem~\cite{df-ectefcgsf-j-09,efln-sldahgcg-06}, that can be
arguably considered as one of the most important open problems in the
field.

The \sefe problem has been proved \NPC for $k\geq 3$ by Gassner
\emph{et al.}~\cite{gjpss-sgefe-06}. On the other hand, if the
embedding of the input graphs is fixed, \sefe becomes polynomial-time
solvable for $k=3$, but remains \NPC for $k \geq
14$~\cite{adf-seepg-j13}.

In Chapter $11$ of the Handbook of Graph Drawing and
Visualization~\cite{bkr-sepg-12}, the \sefe problem with
\emph{sunflower intersection} (\sunsefep) is cited as an open question
(Open Problem $7$). In this setting, the \emph{intersection graph}
\Gint (that is, the graph composed of the edges that are common to at
least two graphs) is such that, if an edge belongs to \Gint, then it
belongs to all the input graphs. Haeupler \emph{et al.}~\cite{hjl-tspcg2c-10} conjectured that \sunsefep is
polynomial-time solvable. However, Schaefer~\cite{s-ttphtpv-13}
recently proved that this problem is \NPC for $k \geq 3$. The
reduction is from the \NPC~\cite{hoske-befpa-12} problem \ptckpbep
(\ptckpbepshort), defined~\cite{adfpr-tsetgibgt-11} as follows. Given
a set $V$ of vertices, a tree $T$ whose leaves are the elements of
$V$, and a collection of edge-sets $E_i\subseteq V\times V$, for
$i=1,\dots,k$, is there a $k$-page book embedding such that the edges
in $E_i$ are placed on the $i$-th page and the ordering of the
elements of $V$ on the spine is represented by $T$?  Note that, the
\NPCN of \ptckpbepshort holds for $k$ unbounded~\cite{hoske-befpa-12},
which implies that the \NPCN of \sunsefep for $k \geq 3$ holds for instances in which
the intersection graph is a spanning forest composed of an unbounded
number of star graphs~\cite{s-ttphtpv-13}.

In this paper, we improve on this result by proving that \sunsefep is
\NPC with $k \geq 3$ even if \Gint consists of a single spanning tree and all the input graphs are biconnected. Note that, for $k=2$, having \Gint connected and all the input graphs biconnected suffices to have a polynomial-time algorithm for the problem~\cite{br-spqacep-13}.

Since \sunsefep when the intersection graph is connected has been proved equivalent to the \ptckpbepshort problem~\cite{adfpr-tsetgibgt-11} (where the equivalence sets each graph $G_i$ equal to tree $T$ plus the edge-set $E_i$), our result implies the \NPCN of \ptckpbepshort for $k\geq 3$, but with no guarantees on the biconnectivity of the input graphs. We prove for this problem even stronger results, namely that \ptckpbepshort remains \NPC for $k\geq 3$ even if two of the input graphs $G_i = T \cup E_i$ are biconnected or if $T$ is a star. This latter setting, in which the tree $T$ basically does not impose any constraint on the ordering of the vertices on the spine, is also known as \pkpbep (\pkpbepshort). Note that, for $k=2$, \pkpbepshort can be solved in linear time~\cite{hn-tpbecgp-09}.

From the algorithmic point of view, we prove that \ptckpbepshort with $k \geq 2$ can be solved in linear time if $k-1$ of the input edge-sets $E_i$ induce connected graphs (a stronger condition than graph $G_i$ being biconnected), hence improving on a result by Hoske~\cite{hoske-befpa-12}, that was based on all the $k$ input edge-sets having this property. Of course, relaxing this constraint on one of the $k$ input edge-sets becomes more relevant for small values of $k$; in particular, it contributes to extend the class of instances that can be solved in polynomial time also for $k=2$, that is the most studied setting both for \ptckpbepshort and for \sefep (note that every instance of \sefep with $k=2$ obviously has sunflower intersection).

In fact, even if the complexity of \sefep and of \ptckpbepshort is still
unknown for $k=2$, polynomial-time algorithms exist
for instances in which:
\begin{inparaenum}[(i)]
\item one of \Gr and \Gb has a fixed
  embedding~\cite{adfjkpr-tppeg-10};
\item the intersection graph \Gint is
  biconnected~\cite{adfpr-tsetgibgt-11,hjl-tspcg2c-10}, a star
  graph~\cite{adfpr-tsetgibgt-11}, or a subcubic
  graph~\cite{hoske-befpa-12,s-ttphtpv-13};
\item each connected component of \Gint has a fixed
  embedding~\cite{br-drpse-13}; or
\item \Gr and \Gb are biconnected and \Gint is
  connected~\cite{br-spqoacep-13}.
\end{inparaenum}

For the setting $k=2$, we also prove that, given any instance of \ptckpbepshort (and hence of \sefep in which \Gint is connected), it is possible to construct an equivalent instance of the same problem in which one of the input graphs, say \Gr, is biconnected and series-parallel. This implies that it would be sufficient to find a polynomial-time algorithm for this seemingly restricted case in order to have a polynomial-time algorithm for the whole problem.

An updated summary of the results on \sunsefep and on \ptckpbepshort is presented in Table~\ref{tb:complexity}.

Still in the setting $k=2$, we study the optimization version of \sefe, that we
call \maxsefep, which is cited as an open question by Haeupler {\em et al.}~\cite{hjl-tspcg2c-10} and
in Chapter $11$ (Open Problem $9$) of the Handbook of Graph Drawing and
Visualization~\cite{bkr-sepg-12}. In this problem, one asks for
drawings of \Gr and \Gb such that as many edges of \Gint as possible
are drawn the same. We prove that \maxsefep is \NPC, even under some
strong constraints.  Namely, the problem is \NPC if \Gr and \Gb are
triconnected, and \Gint is composed of a triconnected component plus a
set of isolated vertices. This implies that the problem is
computationally hard both in the fixed and in the variable embedding
case. In the latter case, however, we can prove that \maxsefep is \NPC
even if \Gint has degree at most $2$. Observe that any of these
constraints would be sufficient to obtain polynomial-time algorithms
for the original decision problem.

In Sect.~\ref{se:preliminaries} we give some preliminary
definitions. In Sect.~\ref{se:sunsefe} we deal with the sunflower
intersection scenario; in Sect.~\ref{se:bookembedding} we focus on the \ptckpbepshort problem; while in Sect.~\ref{se:maxsefe} we study the
\maxsefep problem. Finally, in Sect.~\ref{se:conclusions} we give concluding remarks and discuss some open problems.

\newcommand{\void}{\cellcolor[gray]{0.9}--}
\begin{table}[tb]
  \centering
  \begin{tabular}{|c|c|c|c|c|c|c|}
    \hline
    {\bf Problem} & {$\mathbf G_\cap$} & {\bf T-Coherent} & {$\mathbf{k}$} &  {\bf Biconnected} &  {\bf $\mathcal{T}$-Biconnected} &  {\bf Complexity} \\
    \hline
    {\sc Sunflower} & tree & NO & $k\geq 3$ & $k$ & \void & NPC~(Th.\ref{th:sunsefep}) \\ 
    \cline{1-4}
    \pkpbepshort & star & YES & $k\geq 3$ & \void & \void & NPC~(Th.~\ref{th:p3be})  \\
    \cline{1-4}
    \pTWOpbepshort & star & YES & $k=2$ & \void & \void & $O(n)$~(\cite{hn-tpbecgp-09}) \\
    \cline{1-4}
    \ptcTHREEpbepshort & caterpillar & YES & $k=3$ & $2$ & \void & NPC~(Th.~\ref{th:pt3be-2bico})  \\
    \cline{1-4}
    \ptckpbepshort & tree & YES & $k\geq 2$ & $k-1$ & $k-1$ & $O(n)$~(Th.~\ref{th:algo}) \\
    \cline{1-4}
    & tree & YES & $k=2$ & $2$ & \void & $O(n^2)$~(\cite{br-spqoacep-13}) \\
    \ptcTWOpbepshort & binary tree &  YES & $k=2$ & \void & \void & $O(n^2)$~(\cite{hoske-befpa-12})  \\
    & tree & YES & $k=2$ & $1$ & \void & OPEN~(Th.~\ref{th:seriesparallel}) \\
    \hline
  \end{tabular}
  \vspace{4mm}
  \caption{Complexity status for \sunsefep, \ptckpbepshort, and \pkpbepshort. }\label{tb:complexity}
\end{table}

\section{Preliminaries}\label{se:preliminaries}

A \emph{drawing} of a graph is a mapping of each vertex to a point of
the plane and of each edge to a simple curve connecting its endpoints.
A drawing is \emph{planar} if the curves representing its edges do not
cross except, possibly, at common endpoints.  A graph is \emph{planar}
if it admits a planar drawing.  A planar drawing $\Gamma$ determines a
subdivision of the plane into connected regions, called \emph{faces},
and a clockwise ordering of the edges incident to each vertex, called
\emph{rotation scheme}. The unique unbounded face is the \emph{outer
  face}. Two drawings are \emph{equivalent} if they have the same
rotation schemes. A \emph{planar embedding} is an equivalence class of
planar drawings.

The \sefe problem can be studied both in terms of embeddings and in
terms of drawings, since edges can be represented by arbitrary curves
without geometric restrictions, and since J\"unger and
Schulz~\cite{js-igsefe-09} proved that two graphs \Gr and \Gb with
intersection graph \Gint have a \sefe if and only if there exists a
planar embedding \GammaR of \Gr and a planar embedding \GammaB of \Gb
inducing the same embedding of \Gint. This condition extends to more
than two graphs in the sunflower intersection setting.

A graph is \emph{connected} if every pair of vertices is connected by
a path. A \emph{$k$-connected} graph $G$ is such that removing any
$k-1$ vertices leaves $G$ connected; $3$-connected and $2$-connected
graphs are also called \emph{triconnected} and \emph{biconnected},
respectively. A \emph{tree} is a graph with no cycle. A
\emph{caterpillar} is a tree such that the removal of all the leaves
yields a path. A subgraph $H$ of a graph $G$ is \emph{spanning} if for
each vertex $v \in G$ there exists an edge of $H$ incident to $v$.

A \emph{series-parallel graph (SP-graph)} is a graph with no $K_4$-minor. SP-graphs are inductively defined as follows. An edge $(u,v)$ is an SP-graph with \emph{poles} $u$ and $v$. Denote by $u_i$ and $v_i$ the poles of an SP-graph graph $G_i$. A \emph{series composition} of SP-graphs $G_0,\dots,G_k$, with $k\geq 1$, is an SP-graph with poles $u=u_0$ and $v=v_k$, containing graphs $G_i$ as subgraphs, and such that $v_i=u_{i+1}$, for each $i=0,1,\dots,k-1$. A \emph{parallel composition} of SP-graphs $G_0,\dots,G_k$, with $k\geq 1$, is an SP-graph with poles $u=u_0=u_1=\dots=u_k$ and $v=v_0=v_1=\dots=v_k$ and containing graphs $G_i$ as subgraphs.

The \emph{dual} of a graph $G$ with respect to an embedding $\Gamma$
of $G$ is the graph $G^\star$ having a vertex $v_f$ for each face $f$
of $\Gamma$ and an edge $(v_{f_1},v_{f_2})$ if and only if faces $f_1$
and $f_2$ of $\Gamma$ have a common edge $e$ in $G$. We say that edge
$(v_{f_1},v_{f_2})$ is the \emph{dual edge} of $e$, and vice versa.

\section{\sunsefep}\label{se:sunsefe}

In this section we study the \sunsefep problem, that is the
restriction of \sefep to instances in which the intersection graph
\Gint is the same for each pair of graphs, that is, \Gint $= G_i \cap
G_j$ for each $1 \leq i < j \leq k$. We prove that \sunsefep is \NPC
with $k \geq 3$ even if \Gint is a spanning tree and all the input
graphs are biconnected.

The proof is based on a polynomial-time reduction from the
\NPC~\cite{top-o-79} problem \betp, that takes as input a finite set
$A$ of $n$ objects and a set $C$ of $m$ ordered triples of distinct
elements of $A$, and asks whether a linear ordering $\mathcal{O}$ of
the elements of $A$ exists such that for each triple $\langle
\alpha,\beta,\gamma \rangle$ of $C$, we have either $\mathcal{O} =
<\dots, \alpha, \dots, \beta, \dots,\gamma,\ldots >$ or $\mathcal{O} =
<\dots, \gamma, \dots, \beta, \dots, \alpha, \ldots >$.

In order to simplify the proof, we first give in
Lemma~\ref{le:aux-ptkbe} an \NPCN proof for a less restricted setting
of \sunsefep and then describe how the produced instances can be
modified in order to obtain equivalent instances with the desired
properties.

A \emph{pseudo-tree} is a connected graph containing only one cycle.

\begin{lemma}\label{le:aux-ptkbe}
  \sunsefep with $k= 3$ is \NPC even if two of the input graphs are
  biconnected and the intersection graph \Gint is a spanning
  pseudo-tree.
\end{lemma}
\begin{proof}
  The membership in \NP has been proved in~\cite{gjpss-sgefe-06} by
  reducing \sefe to the \emph{Weak Realizability}
  Problem~\cite{k-cnatg-98,kln-nstl-91}.

  The \NPHN is proved by means of a polynomial-time reduction from
  problem \betp. Given an instance $\langle A,C \rangle$ of \betp, we
  construct an instance \sunsefeinstance of \sunsefep that admits a
  \sefe if and only if $\langle A,C \rangle$ is a positive instance of
  \betp, as follows.

  Refer to Fig.~\ref{fig:sunreduction} for an illustration of the
  construction of \Gint, \Gr, \Gb, and \Gg.

\begin{figure}[htb]
  \centering
  \includegraphics[width=\textwidth]{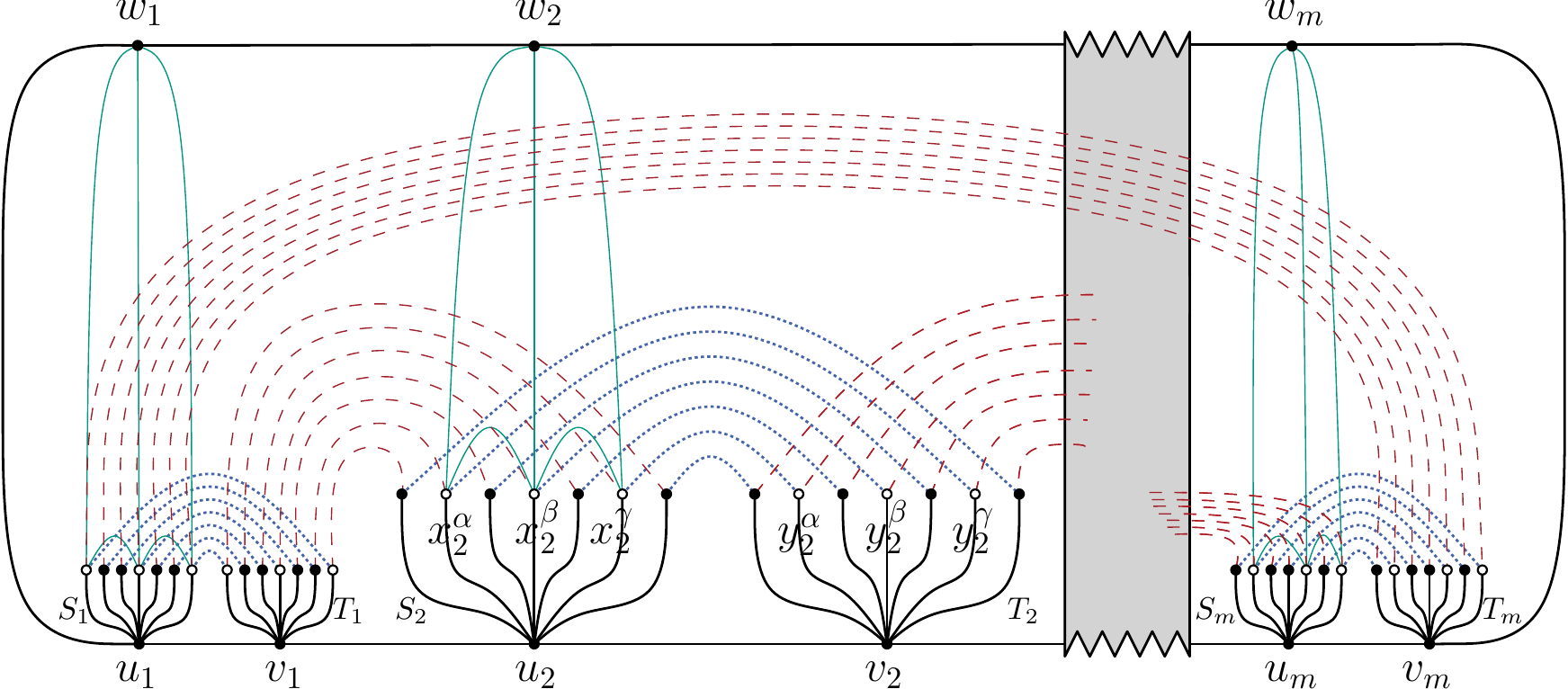}
  \caption{Illustration of the composition of \Gint, \Gr, \Gb, and
    \Gg in Lemma~\ref{le:aux-ptkbe}, focused on the $i$-th triple $t_i=\langle \alpha,\beta,\gamma
    \rangle$ of $C$ with $i=2$.}\label{fig:sunreduction}
\end{figure}

Graph \Gint contains a cycle $\mathcal{C} = u_1, v_1, u_2, v_2,
\dots,$ $u_m, v_m, w_m, \dots, w_1$ of $3m$ vertices. Also, for each
$i=1,\dots,m$, \Gint contains a star $S_i$ with $n$ leaves centered at
$u_i$ and a star $T_i$ with $n$ leaves centered at $v_i$. For each
$i=1, \dots, m$, the leaves of $S_i$ are labeled $x_i^j$ and the
leaves of $T_i$ are labeled $y_i^j$, for $j=1,\dots, n$.
Graph \Gr contains all the edges of \Gint plus a set of edges
$(y_{i}^j, x_{i+1}^j)$, for $i=1, \dots, m$ and $j=1, \dots, n$. Here
and in the following, $i+1$ is computed modulo $m$.
Graph \Gb contains all the edges of \Gint plus a set of edges
$(x_i^j,y_i^j)$, for $i=1, \dots, m$ and $j=1, \dots, n$.
Graph \Gg contains all the edges of \Gint plus a set of edges defined
as follows. For each $i=1, \dots, m$, consider the $i$-th triple
$t_i=\langle \alpha,\beta,\gamma \rangle$ of $C$, and the
corresponding vertices $x_i^\alpha$, $x_i^\beta$, and $x_i^\gamma$ of
$S_i$; graph \Gg contains edges $(w_i, x_i^\alpha)$, $(w_i,
x_i^\beta)$, $(w_i, x_i^\gamma)$, $(x_i^\alpha, x_i^\beta)$, and
$(x_i^\beta, x_i^\gamma)$.

First note that, by construction, \sunsefeinstance is an instance of
\sunsefep, and graph \Gint is a spanning pseudo-tree. Also, one can
easily verify that \Gr and \Gb are biconnected. In the following we
prove that \sunsefeinstance is a positive instance if and only if
$\langle A,C \rangle$ is a positive instance of \betp.

Suppose that \sunsefeinstance is a positive instance, that is, \Gr,
\Gb, and \Gg admit a \sefe \sunsefesolution.
Observe that, for each $i=1,\dots,m$, the subgraph of \Gr induced by
the vertices of $T_i$ and the vertices of $S_{i+1}$ is composed of a
set of $n$ paths of length $3$ between $v_i$ and $u_{i+1}$, where the
$j$-th path contains internal vertices $y_i^j$ and $x_{i+1}^j$, for
$i=1,\dots,n$.  Hence, in any \sefe of \sunsefeinstance, the ordering
of the edges of $T_i$ around $v_i$ is reversed with respect to the
ordering of the edges of $S_{i+1}$ around $u_{i+1}$, where the
vertices of $T_i$ and $S_{i+1}$ are identified based on index $j$.
Also observe that, for each $i=1,\dots,m$, the subgraph of \Gb induced
by the vertices of $S_i$ and the vertices of $T_i$ is composed of a
set of $n$ paths of length $3$ between $u_i$ and $v_i$, where the
$j$-th path contains internal vertices $x_i^j$ and $y_i^j$, for
$i=1,\dots,n$. Hence, in any \sefe of \Gr, \Gb, and \Gg, the ordering
of the edges of $S_i$ around $u_i$ is the reverse of the ordering of
the edges of $T_i$ around $v_i$, where the vertices of $S_i$ and $T_i$
are identified based on $j$.
The two observations imply that, in any \sefe of \Gr, \Gb, and \Gg,
for each $i=1,\dots,m$ the ordering of the edges of $S_i$ around $u_i$
is the same as the ordering of the edges of $S_{i+1}$ around
$v_{i+1}$, where the vertices of $S_i$ and $S_{i+1}$ are identified
based on $j$.

We construct a linear ordering $\mathcal{O}$ of the elements of $A$
from the ordering of the leaves of $S_1$ in
\sunsefesolution. Initialize $\mathcal{O} = \emptyset$; then, starting
from the edge of $S_1$ clockwise following $(u_1,w_1)$ around $u_1$,
consider all the leaves of $S_1$ in clockwise order.  For each
considered leaf $x_1^j$, append $j$ as the last element of
$\mathcal{O}$.
We prove that $\mathcal{O}$ is a solution of $\langle A,C\rangle$.
For each $i=1,\dots,m$, the subgraph of \Gg induced by vertices $w_i$,
$u_i$, $x_i^\alpha$, $x_i^\beta$, and $x_i^\gamma$ is such that adding
edge $(u_i,w_i)$ would make it triconnected. Hence, it admits two
planar embeddings, which differ by a flip. Thus, in any \sefe of \Gr,
\Gb, and \Gg, edges $(u_i,x_i^\alpha)$, $(u_i,x_i^\beta)$, and
$(u_i,x_i^\gamma)$ appear either in this order or in the reverse order
around $u_i$. Since for each triple $t_i=\langle \alpha, \beta, \gamma
\rangle$ in $C$ there exists vertices $w_i$, $u_i$, $x_i^\alpha$,
$x_i^\beta$, and $x_i^\gamma$ inducing a subgraph of \Gg with the
above properties, and since the clockwise ordering of the leaves of
$S_i$ is the same for every $i$, $\mathcal{O}$ is a solution of
$\langle A,C\rangle$.

Suppose that $\langle A,C \rangle$ is a positive instance, that is,
there exists an ordering $\mathcal{O}$ of the elements of $A$ in which
for each triple $t_i$ of $C$, the three elements of $t_i$ appear in
one of their two admissible orderings. We construct an embedding for
\Gr,\Gb, and \Gg. For each $i=1,\dots,m$, the rotation schemes of
$u_i$ and $v_i$ are constructed as follows. Initialize $first=v_{i-1}$
if $i>1$, otherwise $first=w_1$.  Also, initialize $last=u_{i+1}$ if
$i<m$, otherwise $last=w_m$.  For each element $j$ of $\mathcal{O}$,
place $(u_i,x_i^j)$ between $(u_i, first)$ and $(u_i,v_i)$ in the
rotation scheme of $u_i$, and set $first=x_i^j$.  Also, place
$(v_i,x_i^j)$ between $(v_i, last)$ and $(v_i,u_i)$ in the rotation
scheme of $v_i$, and set $last=x_i^j$.
Since all the vertices of \Gr and of \Gb different from $u_i$ and
$v_i$ ($i=1,\dots,m$) have degree $2$, the embeddings \GammaR and
\GammaB of \Gr and \Gb, are completely specified.
To obtain the embedding \GammaG of \Gg, we have to specify the
rotation scheme of $w_i$ and of the three leaves of $S_i$ adjacent to
$w_i$, for $i=1,\dots,m$.  Consider a triple $t_i=\langle \alpha,
\beta, \gamma \rangle$ of $C$. Initialize $first=w_{i-1}$, if $i>1$,
and $first=u_1$ otherwise. Also, initialize $last=w_{i+1}$, if $i<m$,
and $last=v_m$ otherwise. Recall that $\alpha$, $\beta$, and $\gamma$
appear in $\mathcal{O}$ either in this order or in the reverse one. In
the former case, the rotation scheme of $w_i$ is $(w_i, last),(w_i,
x_i^\gamma),(w_i, x_i^\beta),(w_i,x_i^\alpha),(w_i, first)$; the
rotation scheme of $x_i^\alpha$ is $(x_i^\alpha, w_i)$, $(x_i^\alpha,
x_i^\beta),$ $(x_i^\alpha, u_i)$; the rotation scheme of $x_i^\beta$
is $(x_i^\beta, x_i^\alpha),$ $(x_i^\beta, w_i),$ $(x_i^\beta,
x_i^\gamma),$ $(x_i^\beta, u_i)$; and the rotation scheme of
$x_i^\gamma$ is $(x_i^\gamma, x_i^\beta),$ $(x_i^\gamma, w_i),$
$(x_i^\gamma, u_i)$.  In the latter case, the rotation scheme of $w_i$
is $(w_i, last),$ $(w_i,x_i^\alpha),$ $(w_i, x_i^\beta),$
$(w_i,x_i^\gamma),$ $(w_i, first)$; the rotation scheme of
$x_i^\alpha$ is $(x_i^\alpha, x_i^\beta),$ $(x_i^\alpha, w_i),$
$(x_i^\alpha, u_i)$; the rotation scheme of $x_i^\beta$ is
$(x_i^\beta, x_i^\gamma),$ $(x_i^\beta, w_i),$ $(x_i^\beta,
x_i^\alpha),$ $(x_i^\beta, u_i)$; and the rotation scheme of
$x_i^\gamma$ is $(x_i^\gamma, w_i),$ $(x_i^\gamma, x_i^\beta),$
$(x_i^\gamma, u_i)$.
In order to prove that \sunsefesolution is a \sefe, we first observe
that the embeddings of \Gint obtained by restricting \GammaR, \GammaB,
and \GammaG to the edges of \Gint, respectively, coincide by
construction.  The planarity of \GammaR and \GammaB descends from the
fact that the orderings of the edges incident to $u_i$ and $v_i$, for
$i=1,\dots,m$, is one the reverse of the other (where vertices are
identified based on index $j$).  The planarity of \GammaG is due to
the fact that, by construction, for each $i=1,\dots,m$, the subgraph
induced by $w_i$, $u_i$, $x_i^\alpha$, $x_i^\beta$, and $x_i^\gamma$
is planar in \GammaG. This concludes the proof of the theorem.
\end{proof}

We are now ready to prove the main result of the section, by showing
how to modify the reduction of Lemma~\ref{le:aux-ptkbe} to obtain
instances in which all graphs are biconnected and \Gint is a tree.
 
\begin{theorem}\label{th:sunsefep}
  \sunsefep is \NPC for $k\geq 3$ even if all the input graphs are
  biconnected and the intersection graph is a spanning tree.
\end{theorem}

\begin{proof}
  The membership in \NP has been proved in~\cite{gjpss-sgefe-06}.

  The \NPHN is proved by means of a polynomial-time reduction from
  problem \betp. Given an instance $\langle A,C \rangle$ of \betp, we
  first construct an instance \sunsefeinstancep{*} of \sunsefep that
  admits a \sefe if and only if $\langle A,C \rangle$ is a positive
  instance of \betp by applying the reduction shown in
  Lemma~\ref{le:aux-ptkbe}. We show how to modify \sunsefeinstancep{*}
  to obtain an equivalent instance \sunsefeinstance with the required
  properties.

  Refer to Fig.~\ref{fig:sunBICOreduction} for an illustration of the
  construction of \Gint, \Gr, \Gb, and \Gg.

  Graph \Gint is initialized to $\Gintm^*$. For $i=1,\dots,m$,
  subdivide edge $(w_u,w_{i+1})$ (where $w_{m+1}=v_m$) with two
  vertices $s_i$ and $t_i$, add a star with $3$ leaves $\alpha_i$,
  $\beta_i$, and $\gamma_i$ with center $c_i$, and add an edge
  connecting $w_i$ to $c_i$.
  Graph \Gr contains all the edges of \Gint plus a set of edges
  defined as follows. As in \sunsefeinstancep{*}, for $i=1, \dots, m$,
  graph \Gr contains edges $(y_{i}^j, x_{i+1}^j)$, with $j=1, \dots,
  n$, connecting the leaves of $T_i$ to the leaves of
  $S_{i+1}$. Additionally, for $i=1, \dots, m$, \Gr contains edges
  $(w_i,\alpha_i)$,$(\alpha_i,\beta_i)$,$(\beta_i,\gamma_i)$,
  $(\gamma_i,w_i)$, and $(\beta_i,s_i)$. Here and in the following,
  $i+1$ is computed modulo $m$.
  Graph \Gb contains all the edges of \Gint plus a set of edges
  defined as follows.  As in \sunsefeinstancep{*}, for $i=1, \dots,
  m$, graph \Gb contains edges $(x_i^j,y_i^j)$, with $j=1, \dots,
  n$. Additionally, for $i=1, \dots, m$, \Gb contains edges
  $(\alpha_i,t_i)$, $(\beta_i,t_i)$, and $(\gamma_i,t_i)$.
  Graph \Gg contains all the edges of \Gint plus a set of edges
  defined as follows.  For each $i=1, \dots, m$, consider the $i$-th
  triple $t_i=\langle \alpha,\beta,\gamma \rangle$ of $C$, and the
  corresponding vertices $x_i^\alpha$, $x_i^\beta$, and $x_i^\gamma$
  of $S_i$; graph \Gg contains edges $(\alpha_i, x_i^\alpha)$,
  $(\beta_i, x_i^\beta)$, $(\gamma_i, x_i^\gamma)$, and edges
  $(x_i^j,c_i)$, for every $j \notin \{\alpha,\beta,\gamma\}$.  Also,
  for $i=1, \dots, m$, graph \Gg contains edges $(y_i^j,t_i)$, with
  $j=1,\dots,n$.

  Observe that, graph \Gint is a pseudo-tree and graphs \Gr, \Gb, and
  \Gg are biconnected.  We first prove that the constructed instance
  \sunsefeinstance of \sunsefep is equivalent to instance $\langle A,C
  \rangle$ of \betp. Then, we show how to modify \sunsefeinstance in
  such a way that \Gint is a tree, without losing the biconnectivity
  of the input graphs.

  Suppose that \sunsefeinstance is a positive instance, that is, \Gr,
  \Gb, and \Gg admit a \sefe \sunsefesolution.
  Observe that, as proved in Lemma~\ref{le:aux-ptkbe} for
  \sunsefeinstancep{*}, in any \sefe of \Gr, \Gb, and \Gg, for each
  $i=1,\dots,m$, the ordering of the edges of $S_i$ around $u_i$ is
  the same as the ordering of the edges of $S_{i+1}$ around $v_{i+1}$,
  where the vertices of $S_i$ and $S_{i+1}$ are identified based on
  index $j$.

  We construct a linear ordering $\mathcal{O}$ of the elements of $A$
  from the ordering of the leaves of $S_1$ in \sunsefesolution as
  described in Lemma~\ref{le:aux-ptkbe}.

  We prove that $\mathcal{O}$ is a solution of $\langle A,C\rangle$.
  For each $i=1,\dots,m$, the subgraph of \Gr induced by vertices
  $w_i$, $\alpha_i$, $\beta_i$, $\gamma_i$, and $c_i$ is a
  triconnected subgraph attached to the rest of the graph through the
  split pair $\{w_i,\beta_i\}$. Hence, in any planar embedding of \Gr
  (and hence also in \GammaR) the clockwise order of the edges around
  $c_i$ is either $(c_i,\alpha_i)$, $(c_i,w_i)$, $(c_i,\gamma_i)$, and
  $(c_i,\beta_i)$, or $(c_i,\alpha_i)$, $(c_i,\beta_i)$,
  $(c_i,\gamma_i)$, and $(c_i,w_i)$.  Also, the ordering of the edges
  of \Gg around $c_i$ in \GammaR restricted to those belonging to
  \Gint is the same as in \GammaR. Further, for each $i=1,\dots,m$,
  consider the subgraph of \Gg composed of the paths connecting $c_i$
  and $u_i$, and containing a leaf of $S_i$. In any planar embedding
  of \Gg the ordering of the edges around $c_i$ is reversed with
  respect to the ordering of the edges around $u_i$, where the edges
  are identified based on the path they belong to. Hence, the ordering
  of the edges of $S_i$ around $u_i$ (and thus $\mathcal{O}$) is such
  that edges $(u_i,x_i^\alpha)$, $(u_i,x_i^\beta)$, and
  $(u_i,x_i^\gamma)$ appear either in this order or in the reverse
  order.  Since the clockwise ordering of the edges of $S_i$ around
  $u_i$ is the same for every $i$, $\mathcal{O}$ is a solution of
  $\langle A,C\rangle$.

  Suppose that $\langle A,C \rangle$ is a positive instance, that is,
  there exists an ordering $\mathcal{O}$ of the elements of $A$ in
  which for each triple $t_i$ of $C$ the three elements of $t_i$
  appear in one of their two admissible orderings. We construct
  embeddings \GammaR,\GammaB, and \GammaG for \Gr,\Gb, and \Gg,
  respectively. For each $i=1,\dots,m$, the rotation schemes of $u_i$
  and $v_i$ in \GammaR, in \GammaB, and in \GammaG are constructed
  based on $\mathcal{O}$ as described in the proof of
  Lemma~\ref{le:aux-ptkbe}.
  Note that, in any \sefe of \sunsefeinstance all the vertices not
  belonging to the only cycle of \Gint lie on the same side with
  respect to it, as removing such a cycle from the union graph \Gun
  results in a connected graph. Hence, the rotation scheme of $w_i$
  restricted to the edges of \Gint is determined in \GammaR, \GammaB,
  and \GammaG. Also, the rotation scheme of $s_i$ is determined in
  \GammaR.
  Consider the $i$-th triple $t_i=\langle \alpha,\beta,\gamma \rangle$
  of $C$. We set the rotation scheme of $c_i$ restricted to the edges
  of \Gint in \GammaR, \GammaB, and \GammaG to be either $(c_i,w_i)$,
  $(c_i,\gamma_i)$, $(c_i,\beta_i)$, and $(c_i,\alpha_i)$, if
  $\alpha$, $\beta$, and $\gamma$ appear in this order in
  $\mathcal{O}$, or $(c_i,w_i)$, $(c_i,\alpha_i)$, $(c_i,\beta_i)$,
  and $(c_i,\gamma_i)$, if they appear in the reverse order in
  $\mathcal{O}$. Note that, given the rotation scheme of $c_i$ in
  \GammaR and in \GammaB, the rotations schemes of $w_i$, $\alpha_i$,
  $\beta_i$, and $\gamma_i$ in \GammaR and of $t_i$ in \GammaB are
  univocally determined. Observe that \GammaR and \GammaB are planar
  by construction.  We prove that \GammaG can be completed to a planar
  drawing of \Gg.  In order to do that, we need to specify the
  rotation schemes of $c_i$ and $t_i$ in \GammaG.  We set the rotation
  schemes of $c_i$ and of $t_i$ to be the reverse with respect to the
  rotation schemes of $u_i$ and of $v_i$, respectively, where edges
  are identified based on the path they belong to.  As for $t_i$, this
  clearly does not introduce crossings in \GammaG, while for $c_i$
  this is due to the fact that the ordering of the edges of \Gint
  incident to $c_i$ determined by the $i$-th triple is consistent with
  the rotation scheme of $u_i$, since this has been determined by
  $\mathcal{O}$.

  In order to prove that \sunsefesolution is a \sefe, we observe that
  the embeddings of \Gint obtained by restricting \GammaR, \GammaB,
  and \GammaG to the edges of \Gint, respectively, coincide by
  construction.

  Finally, in order to make \Gint a spanning tree, remove edge
  $(u_1,w_1)$ from \Gint; add to \Gint two star graphs with $3$
  leaves, and add to \Gint an edge connecting $u_1$ to the center of
  the first star and an edge connecting $w_1$ to the center of the
  second star. Also, add edges to \Gr, to \Gb, and to \Gg among
  vertices of the two stars so that (i) all graphs remains
  biconnected, (ii) there exists an edge of \Gr, an edge of \Gb, and
  an edge of \Gg connecting a leaf of the first star to a leaf of the
  second star, and (iii) no edge is added to more than one graph. A
  suitable augmentation is shown in Fig.~\ref{fig:sunBICOreduction}.

\begin{figure}[htb]
  \centering
  \includegraphics[width=\textwidth]{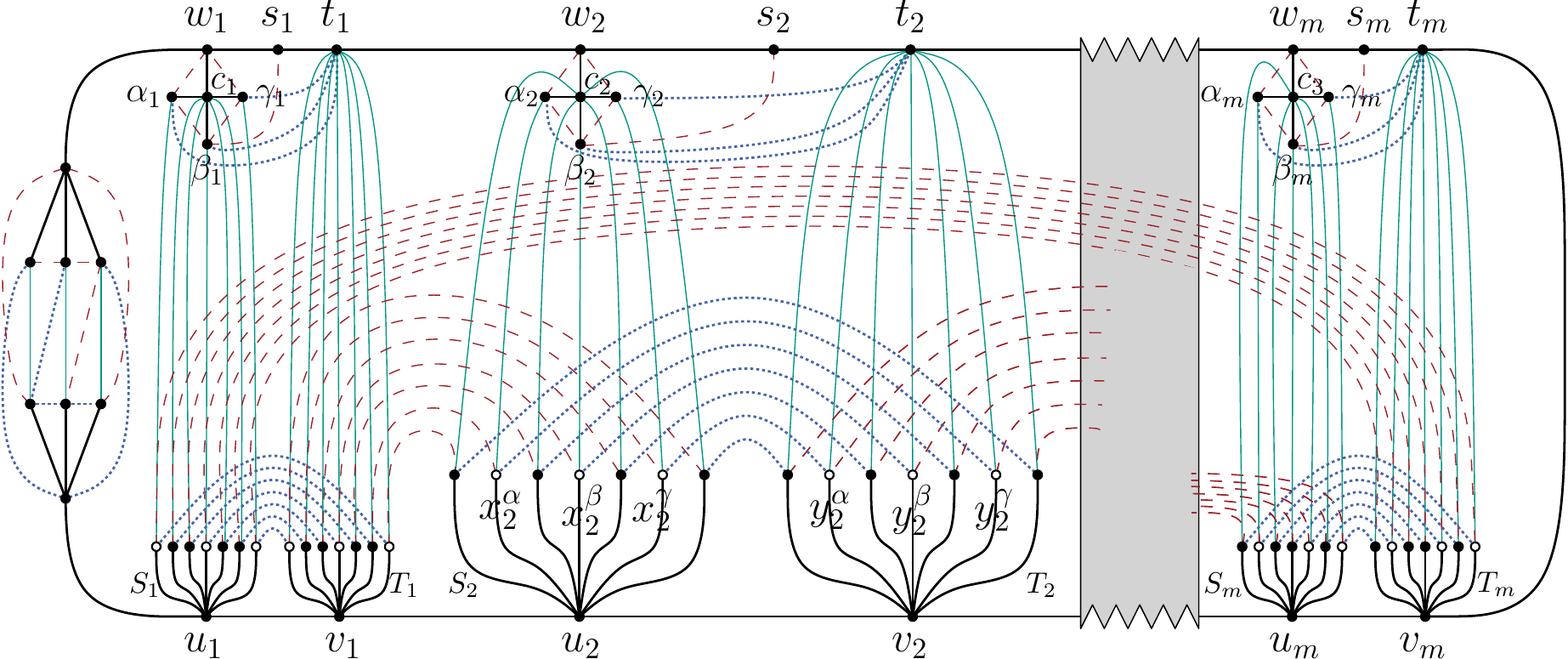}
  \caption{Illustration of the composition of \Gint, \Gr, \Gb, and
    \Gg in Theorem~\ref{th:sunsefep}, focused on the $i$-th triple $t_i=\langle \alpha,\beta,\gamma
    \rangle$ of $C$ with $i=2$.}\label{fig:sunBICOreduction}
\end{figure}

The above discussion proves the statement for $k=3$. To extend the
theorem to any value of $k$ observe that, given an instance of
\sunsefep with $k_0\geq 3$ biconnected graphs whose intersection graph
\Gint is a tree, an equivalent instance with $k_0 + 1$ biconnected
graphs whose intersection graph is a tree can be obtained by
subdividing an edge of \Gint with a dummy vertex and by connecting it
to all the leaves of \Gint with edges only belonging to the
$(k_0+1)$-th graph.
\end{proof}

\section{\pkpbep}\label{se:bookembedding}

In this section we turn our attention to the problem of computing
$k$-page book embeddings in which the assignment of the $k$ sets of
edges to the $k$ pages is given as part of the input. We study this
problem both in its original definition~\cite{hn-tpbecgp-09}, called \pkpbep
(\pkpbepshort), and in a generalization of it, called \ptckpbep (\ptckpbepshort),
in which the order of the vertices on the spine must satisfy an
additional constraint, namely it must be represented by a tree $T$,
also given as part of the input. Observe that, problem \ptckpbepshort
in which $T$ is a star is exactly the same problem as \pkpbepshort.

Problem \ptckpbepshort has been defined in \cite{adfpr-tsetgibgt-11}
and proved equivalent to the case of \sunsefep in which the
intersection graph \Gint is a spanning tree and all the edges not
belonging to \Gint are incident to two leaves of such
tree\footnote{Although~\cite{adfpr-tsetgibgt-11} proves the
  equivalence for $k=2$, the result can be naturally extended to any
  $k$.}. For this reason, in the following we will indifferently denote an instance
\ptckpbeinstance of \ptckpbepshort by the corresponding instance
$\langle G_1, \dots, G_k \rangle$ of \sunsefep, where $G_i=(V(T),E(T)
\cup E_i)$, for each $i=1,\dots,k$, and vice versa.

We remark that the instances of \sunsefep constructed in the reduction
performed in Theorem~\ref{th:sunsefep} are such that the intersection
graph \Gint is a spanning tree, but there exist edges not belonging to
\Gint that are incident to internal vertices of such tree. In order to obtain
equivalent instances of \sunsefep satisfying both properties, it would
be possible to apply a procedure described
in~\cite{adfpr-tsetgibgt-11} that, for each edge $e \in
\bigcup_{i=1}^k E_i$ incident to an internal vertex $v$ of \Gint, adds
a new leaf to \Gint attached to $v$ and replaces $v$ with this leaf as
an endvertex of $e$. Hence, Theorem~\ref{th:sunsefep} implies that
\ptckpbepshort is \NPC for $k \geq 3$.  However, every time a new leaf
is attached to an internal vertex, such a vertex becomes a cut-vertex
for $k-1$ of the input graphs; thus, none of the $k$ graphs $G_i$ can be
assumed to be biconnected after the whole procedure has been applied.

The relevance of this latter observation is motivated by the fact that
the biconnectivity of the input graphs $G_i$, together with the
``simplicity'' of $T$, seems to be the key factor that allows for
polynomial-time algorithms for the partitioned book embedding
problems.  Indeed, Hoske~\cite{hoske-befpa-12} proved that
\pkpbepshort becomes solvable in linear-time if each graph $G_i$ is
\emph{$T$-biconnected}, that is, $E_i$ induces a connected
graph. Notice that, $T$-biconnectivity is a stronger requirement than
biconnectivity, since the former implies the latter, while the
converse does not always hold.  We observe that the algorithm by Hoske
can be easily generalized from \pkpbepshort to \ptckpbepshort in which
$T$ is not necessarily a star; hence, the same algorithmic result can
be stated also for \ptckpbepshort.  Furthermore, to support the
importance of the above mentioned key factors, we recall that
\ptckpbepshort is polynomial-time solvable for $k=2$ if either both
input graphs are biconnected~\cite{br-spqacep-13}, or $T=\Gintm$ is a
star~\cite{hn-tpbecgp-09}, or $T=\Gintm$ is a binary
tree~\cite{hoske-befpa-12,s-ttphtpv-13}.

In this section we provide several results that considerably narrow the gap
between the instances of the partitioned book embedding problems that can be
solved in polynomial time and those that cannot (unless $P=NP$), by
studying their complexity with respect to such factors. Namely, we
prove that:
\begin{itemize}[$\circ$]
\item \ptckpbepshort remains \NPC for $k=3$ when $T$ is a caterpillar
  and $2$ of the input graphs are biconnected
  (Theorem~\ref{th:pt3be-2bico});
\item \pkpbepshort (with no
restriction on the biconnectivity) is \NPC for $k\geq 3$ (Theorem~\ref{th:p3be}), which was known only for $k$ unbounded~\cite{s-ttphtpv-13}; 
\item \ptckpbepshort is linear-time solvable if $k-1$ of the input
  graphs are $T$-biconnected (Theorem~\ref{th:algo});
\item requiring one of the two graphs of an instance
  \ptcTWOpbeinstance of \ptcTWOpbepshort to be biconnected (and even
  series-parallel) does not alter the computational complexity of the
  problem (Theorem~\ref{th:seriesparallel}).
\end{itemize}

Due to the equivalence between \ptckpbepshort and \sunsefep in which
\Gint is a spanning tree and all the edges not belonging to \Gint
connect two of its leaves, in order to prove
Theorem~\ref{th:pt3be-2bico} it suffices to show that the instances
produced in the reduction of Lemma~\ref{le:aux-ptkbe} can be modified
to obtain equivalent instances satisfying the above properties in
which two of the input graphs are biconnected.

\begin{theorem}\label{th:pt3be-2bico}
  \ptckpbepshort is \NPC for $k=3$ even if two of the input graphs are
  biconnected and $T=\Gintm$ is a caterpillar tree.
\end{theorem}
\begin{proof}
  Consider an instance \sunsefeinstance obtained from the reduction
  described in Lemma~\ref{le:aux-ptkbe}. We describe how to obtain an
  equivalent instance satisfying the required properties.

  Refer to Fig.~\ref{fig:sunreduction} and to
  Fig.~\ref{fig:tc3pbe}. First, for $i=1,\dots,m$, replace the edges
  $(w_i,x_i^\alpha)$, $(w_i,x_i^\beta)$, and $(w_i,x_i^\gamma)$ of \Gg
  with length-$2$ paths composed of a black and of a green edge and
  such that the black edge is incident to $w_i$. Denote by $\Phi_i$
  the star graph centered at $w_i$ induced by the newly inserted black
  edges. Second, for $i=1,\dots,m$, subdivide edge $(w_i,w_{i+1})$ of
  \Gint (where $w_{m+1}=v_m$) with a dummy vertex $d_i$, and add to
  \Gint a star graph $\Psi_i$ centered at $d_i$ and with $3$
  leaves. Observe that, at this stage of the construction, \Gint is a
  spanning pseudo-caterpillar.

\begin{figure}[htb]
  \centering
  \includegraphics[width=\textwidth]{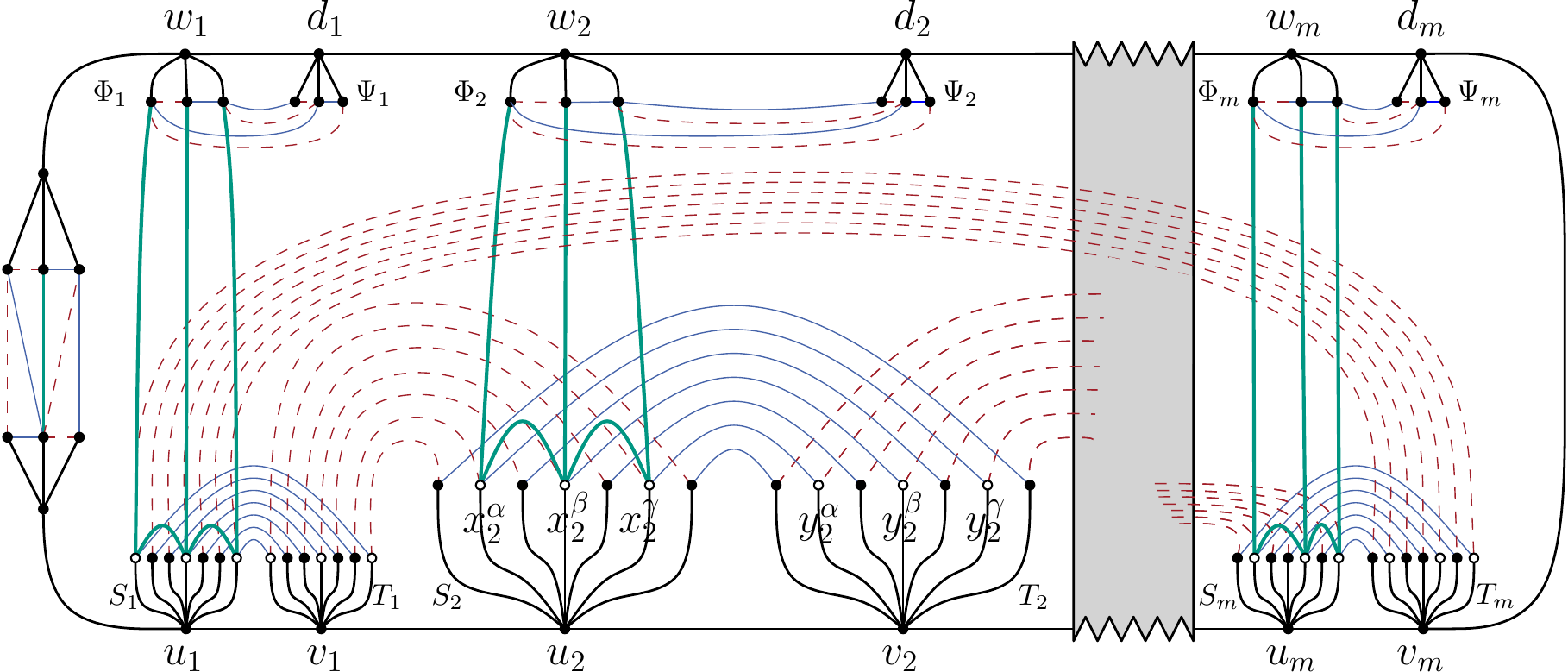}
  \caption{Illustration of how to modify the instance of \sunsefep so
    that: (i) the intersection graph \Gint is a spanning caterpillar
    and (ii) \Gr and \Gb are biconnected.}\label{fig:tc3pbe}
\end{figure}

It is now possible to obtain an equivalent instance of \sunsefep where
\Gr and \Gb are biconnected and \Gint remains a spanning
pseudo-caterpillar, by only adding edges to \Gr and to \Gb among the
leaves of $\Phi_i$ and $\Psi_i$, for $i=1,\dots,m$.

Further, in order to make \Gint a spanning caterpillar, remove edge
$(u_1,w_1)$ from \Gint; add to \Gint two star graphs with $3$ leaves,
and add to \Gint an edge connecting $u_1$ to the center of the first
star and an edge connecting $w_1$ to the center of the second star.

Finally, add edges to \Gr, to \Gb, and to \Gg among the leaves of the
two stars so that (i) \Gr and \Gb are biconnected, (ii) there exists
an edge of \Gg connecting a leaf of the first star to a leaf of the
second star, and (iii) no edge is added to more than one graph.  A
suitable augmentation is shown in Fig.~\ref{fig:tc3pbe}.

It is easy to observe that the constructed instance satisfies the
required properties.
\end{proof}

In the following we prove that dropping the requirement of
biconnectivity of the graphs allows us to prove \NPCN also for
\pkpbepshort when $k$ is bounded by a constant, thus improving on the
result of Hoske~\cite{hoske-befpa-12}. We first prove that the \NPCN
of \ptckpbepshort for $k\geq 3$ proved in Theorem~\ref{th:pt3be-2bico}
implies the \NPCN of \pkpbepshort for $k\geq 4$. Then, in
Theorem~\ref{th:p3be} we show that
\pkpbepshort is \NPC even for $k=3$. We recall that a linear-time
algorithm for the problem is known when
$k=2$~\cite{hn-tpbecgp-09}.

\begin{theorem}\label{th:pbek1}
  \ptckpbepshort is polynomial-time reducible to {\scshape PBE-$(k+1)$}.
\end{theorem}

\begin{proof}
  Let \ptckpbeinstance be an instance of \ptckpbepshort. We construct
  an instance $\langle V^*, E^*_1,\dots, E^*_k, E^*_{k+1} \rangle$ of
  {\scshape PBE-$(k+1)$} as follows.

  Set $V^* = V(T)$ and $E^*_{k+1}=E(T)$. Then, for each $i = 1, \dots,
  k$, set $E^*_i = E_i$. Refer to Fig.~\ref{fig:kPLUS1book}.

\begin{figure}[htb]
  \centering
  \subfigure[\ptcTWOpbeinstance]{\includegraphics[height=0.195\textwidth]{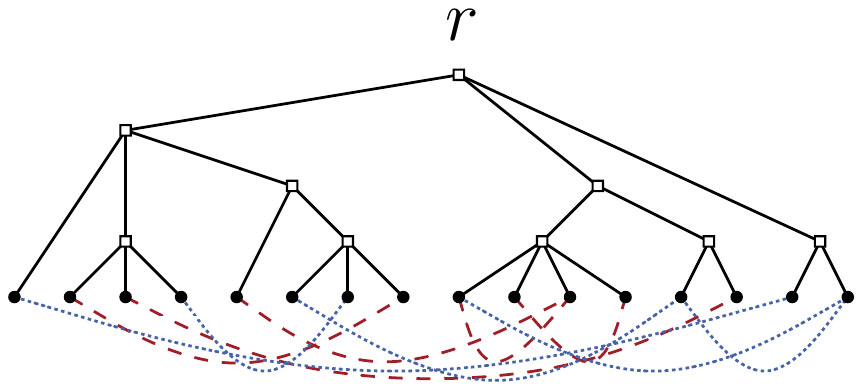}}
  \subfigure[$\langle V^*, \red{E^*_1},\blue{E^*_2},\green{E^*_3}
  \rangle$]{
    \includegraphics[height=0.195\textwidth]{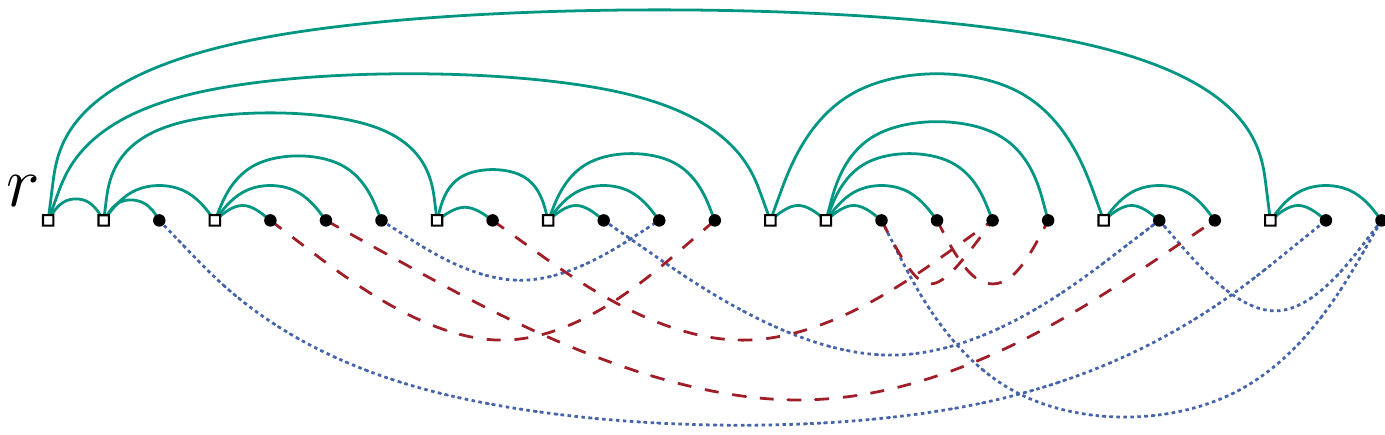}}
  \caption{Illustration of the proof of
    Theorem~\ref{th:pbek1}.}\label{fig:kPLUS1book}
\end{figure}

We prove that $\langle V^*, E^*_1,\dots, E^*_k, E^*_{k+1} \rangle$ is
a positive instance of {\scshape PBE-$(k+1)$} if and only if
\ptckpbeinstance is a positive instance of \ptckpbepshort.

Suppose that $\langle V^*, E^*_1,\dots, E^*_k, E^*_{k+1} \rangle$
admits a partitioned $(k+1)$-page book embedding $\mathcal{O}^*$. Let
$\mathcal{O}$ be the order obtained by restricting $\mathcal{O}^*$ to
the leaves of $T$. We show that $\mathcal{O}$ is a partitioned
$T$-coherent $k$-page book embedding of \ptckpbeinstance.

For each $i=1,\dots,k$, no two edges of $E_i$ alternate in
$\mathcal{O}$, as otherwise the corresponding two edges of $E_i^*$
would alternate in $\mathcal{O}^*$, hence contradicting the hypothesis
that $\mathcal{O}^*$ is a partitioned $(k+1)$-page
book embedding. Also, we claim that order $\mathcal{O}$ is represented
by $T$. Namely, place the vertices of $T$ on a horizontal line in the
same order as they appear in $\mathcal{O}^*$; since $\mathcal{O}^*$
supports a crossing-free drawing of the edges of $E_{k+1}^* = E(T)$ on
a single page and since $\mathcal{O}^*$ restricted to the leaves of
$T$ coincides with $\mathcal{O}$, the claim follows.

Suppose that \ptckpbeinstance admits a partitioned $T$-coherent
$k$-page book embedding $\mathcal{O}$. We show how to construct a
partitioned $(k+1)$-page book embedding $\mathcal{O}^*$ of $\langle
V^*, E^*_1,\dots, E^*_k, E^*_{k+1} \rangle$.

Initialize $\mathcal{O}^* = \mathcal{O}$. Root $T$ at an arbitrary
internal vertex. Then, consider each internal vertex $w$ of $T$
according to a bottom-up traversal.  Consider the subtree $T(w)$ of
$T$ rooted at $w$ and consider the vertex $z$ of $T(w)$ appearing in
$\mathcal{O}^*$ right before all the other vertices of $T(w)$. Place
$w$ right before $z$ in $\mathcal{O}^*$.

We show that $\mathcal{O}^*$ is a partitioned $(k+1)$-page
book embedding of $\langle V^*, E^*_1,\dots, E^*_k, E^*_{k+1}
\rangle$.

For each $i=1,\dots,k$, no two edges of $E_i^*$ alternate in
$\mathcal{O}^*$, as otherwise the corresponding two edges of $E_i$
would alternate in $\mathcal{O}$, hence contradicting the hypothesis
that $\mathcal{O}$ is a partitioned $T$-coherent $k$-page
book embedding. Also, the fact that no two edges of $E_{k+1}^*$
alternate in $\mathcal{O}^*$ descends from the fact that, for each
vertex $w$ of $T$, all the vertices belonging to the subtree $T(w)$ of
$T$ rooted at $w$ appear consecutively in $\mathcal{O}^*$. We prove
this property by induction. In the base case $w$ is the parent of a
set of leaves. In this case, the statement holds since $\mathcal{O}$
is represented by $T$. Inductively assume that, for all children $u_i$
of $w$, the vertices of $T(u_i)$ are consecutive in
$\mathcal{O}^*$. Also, by construction, $w$ has been placed right
before all vertices of $T(w)$. It follows that all vertices of $T(w)$
(including $w$) are consecutive in $\mathcal{O}^*$.  This concludes
the proof of the theorem.
\end{proof}

As \pkpbepshort is a special case of \ptckpbepshort, the problem belongs to \NP. Hence, putting together the results of Theorem~\ref{th:pbek1} and of Theorem~\ref{th:pt3be-2bico}, we obtain the following:

\begin{corollary}\label{co:book-4}
  \pkpbepshort is \NPC for $k\geq 4$.
\end{corollary}

We strengthen this result by proving that the \NPHN of \pkpbepshort
holds even for $k=3$. As for Theorem~\ref{th:pt3be-2bico}, we describe
the proof in terms of the corresponding \sunsefep problem, namely in
the case in which \Gint is a star graph and all the edges not
belonging to \Gint connect two of its leaves.

\begin{theorem}\label{th:p3be}
  \pkpbepshort is \NPC for $k\geq 3$.
\end{theorem}

\begin{proof}
  We prove the statement for $k=3$, as for $k\geq 4$ it descends from Corollary~\ref{co:book-4}.
  The \NPHN is shown by means of a polynomial-time reduction from
  problem \betp. Given an instance $\langle A,C \rangle$ of \betp, we
  construct an instance $\langle V,
  \red{E_1},\blue{E_2},\green{E_3}\rangle$ of \pTHREEpbepshort that
  admits a partitioned $3$-page book embedding if and only if $\langle
  A,C \rangle$ is a positive instance of \betp.

  We describe instance $\langle V,
  \red{E_1},\blue{E_2},\green{E_3}\rangle$ in terms of the
  corresponding instance \sunsefeinstancep{} of \sunsefep in which
  \Gint is a star. Refer to Fig.~\ref{fig:3book}.

\begin{figure}[htb]
  \centering
  \includegraphics[width=\textwidth]{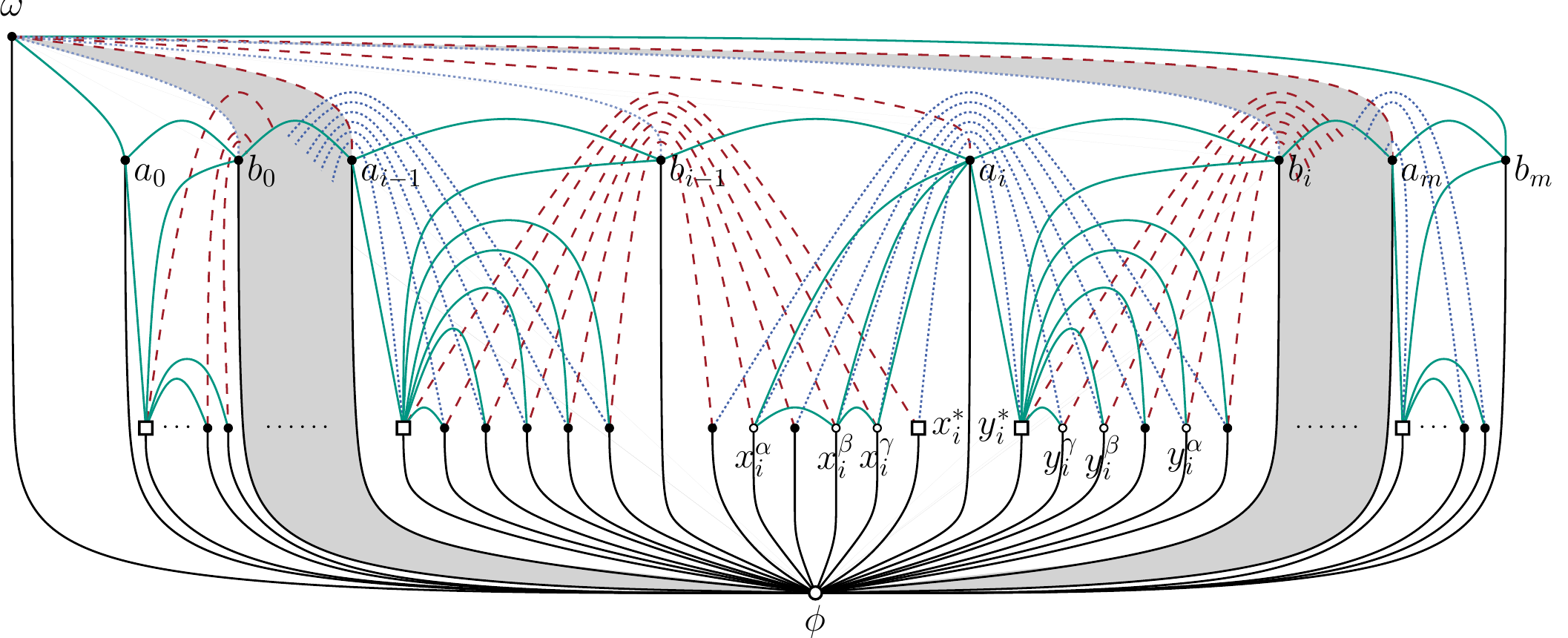}
  \caption{Illustration of the composition of \Gint, \Gr, \Gb, and
    \Gg in Theorem~\ref{th:p3be}, focused on the $i$-th triple $t_i=\langle \alpha,\beta,\gamma
    \rangle$ of $C$.}\label{fig:3book}
\end{figure}

Graph \Gint is initialized to a star graph with center $\phi$, a leaf
$\omega$ and, for $i=0,\dots,m$, leaves $a_i$ and $b_i$.  Also, for
$i=1,\dots,m$, \Gint contains $n$ leaves $x^1_i,\dots,x^n_i$, $n$
leaves $y^1_i,\dots,y^n_i$, plus two additional leaves $x^*_i$ and
$y^*_i$. Finally, \Gint contains $n$ leaves $y^1_0,\dots,y^n_0$, plus
an additional leaf $y^*_0$.

Graph \Gr contains all the edges of \Gint plus a set of edges defined
as follows. For $i=1, \dots, m$, graph \Gr contains an edge $(\omega,
a_i)$. Also, for $i=1, \dots, m$, graph \Gr contains edges
$(x^j_i,y^j_{i-1})$, with $j=1, \dots n$, and edge
$(x^*_i,y^*_{i-1})$.

Graph \Gb contains all the edges of \Gint plus a set of edges defined
as follows. For $i=0, \dots, m-1$, graph \Gb contains an edge
$(\omega, b_i)$. Also, for $i=1, \dots, m$, graph \Gb contains edges
$(x^j_i,y^j_{i})$, with $j=1, \dots n$, and edge $(x^*_i,y^*_{i})$.

Graph \Gg contains all the edges of \Gint plus a set of edges defined
as follows. Graph \Gg contains edges $(\omega, a_o)$ and
$(\omega,b_m)$. Also, for each $i=0, \dots, m$, graph \Gg contains
edges $(a_i,b_i)$, $(a_i,y^*_i)$, $(b_i,y^*_i)$, and edges
$(y^*_i,x^j_i)$, with $j=1,\dots,n$.  Finally, for $i=1,dots,m$,
consider the $i$-th triple $t_i=\langle \alpha,\beta,\gamma \rangle$
of $C$, and the corresponding vertices $x_i^\alpha$, $x_i^\beta$, and
$x_i^\gamma$; graph \Gg contains edges $(a_i, x_i^\alpha)$, $(a_i,
x_i^\beta)$, $(a_i, x_i^\gamma)$, $(x_i^\alpha,x_i^\beta)$, and
$(x_i^\beta,x_i^\gamma)$.

We prove that the constructed instance \sunsefeinstance of \sunsefep
is equivalent to instance $\langle A,C \rangle$ of \betp.

Suppose that \sunsefeinstance is a positive instance, that is, \Gr,
\Gb, and \Gg admit a \sefe \sunsefesolution.
Observe that, vertices $\phi$, $\omega$, and vertices $a_i$ and $b_i$,
with $i=1,\dots,m$, induce a wheel with central vertex $\phi$ in \Gg.
Hence, in any planar embedding of \Gg, edges $(\omega,\phi)$,
$(a_o,\phi)$, $(b_o,\phi)$,$\dots$,$(a_m,\phi)$, and $(b_m,\phi)$
appear in this order (or in the reverse order) around $\phi$.  Also,
since $y^*_i$ is adjacent in \Gg to both $a_i$ and $b_i$, for
$i=0,\dots,m$, edge $(y^*_i,\phi)$ appears between edges $(a_i,\phi)$
and $(b_i,\phi)$ around $\phi$ in any planar embedding of \Gg. Hence,
since all vertices $y^j_i$, with $j=1,\dots,n$, are adjacent in \Gg to
$y^*_i$, also edges $(y^j_i,\phi)$ appear between $(a_i,\phi)$ and
$(b_i,\phi)$ around $\phi$ in any planar embedding of \Gg.
Furthermore, for $i=1,\dots,m$, edges $(x^j_i,\phi)$, with
$j=1\dots,n$, and edge $(x^*_i,\phi)$ appear between $(b_{i-1},\phi)$
and $(a_i,\phi)$ around $\phi$ in \sunsefesolution. This is due to the
following two facts: (1) all vertices $x^j_i$ and vertex $x^*_i$ are
adjacent in \Gr to a vertex $y_{i-1}$ such that edge $(y_{i-1},\phi)$
appears between edges $(a_{i-1},\phi)$ and $(b_{i_1},\phi)$ around
$\phi$, and in \Gb to a vertex $y_i$ such that edge $(y_i,\phi)$
appears between edges $(a_{i},\phi)$ and $(b_{i},\phi)$ around $\phi$;
(2) there exists edges $(\omega,b_{i-1})$ in \Gb and $(\omega,a_i)$ in
\Gr. Refer to Fig.~\ref{fig:3book} for a possible ordering of the
edges around $\phi$ in a \sefe.

Observe that, due to the properties of the ordering of the edges of
\Gint around $\phi$ discussed above, for $i=1,\dots,m$, edge
$(x^*_i,\phi)$ and edges $(x^j_i,\phi)$, with $j=1,\dots,n$, behave
similarly to the edges of the star graph $S_i$ used in
Lemma~\ref{le:aux-ptkbe}, and edge $(y^*_i,\phi)$ and edges
$(y^j_i,\phi)$, with $j=1,\dots,n$, behave similarly to the edges of
the star graph $T_i$ used in Lemma~\ref{le:aux-ptkbe}. Namely, in any
\sefe of \Gr, \Gb, and \Gg, for each $i=1,\dots,m-1$, the ordering of
the edges $(x^j_i,\phi)$, with $j=1,\dots,n$, and edge $(x^*_i,\phi)$
around $\phi$ is the same as the ordering of the edges
$(x^j_{i+1},\phi)$, with $j=1,\dots,n$, and edge $(x^*_{i+1},\phi)$
around $\phi$, where the vertices are identified based on index $j$.

We construct a linear ordering $\mathcal{O}$ of the elements of $A$
from the ordering of the leaves of $x_1^j$, with $j=1,\dots,n$, in
\sunsefesolution as described in Lemma~\ref{le:aux-ptkbe}.

We prove that $\mathcal{O}$ is a solution of $\langle A,C\rangle$.
For each $i=1,\dots,m$, the subgraph of \Gg induced by vertices
$\phi_i$, $x^\alpha_i$, $x^\beta_i$, $x^\gamma_i$, and $a_i$ is a
triconnected subgraph attached to the rest of the graph through the
split pair $\{\phi,a_i\}$. Hence, in any planar embedding of \Gg (and
hence also in \GammaG) edges $(\phi,x^\alpha_i)$, $(\phi,x^\beta_i)$,
$(\phi,x^\gamma_i)$ appear either in this order or in the reverse
order around $\phi$. Since the ordering of the edges
$(x^j_{i+1},\phi)$, with $j=1,\dots,n$, around $\phi$ is the same for
every $i$, $\mathcal{O}$ is a solution of $\langle A,C\rangle$.

Suppose that $\langle A,C \rangle$ is a positive instance, that is,
there exists an ordering $\mathcal{O}$ of the elements of $A$ in which
for each triple $t_i$ of $C$ the three elements of $t_i$ appear in one
of their two admissible orderings. In order to construct embeddings
\GammaR,\GammaB, and \GammaG for \Gr,\Gb, and \Gg, respectively, we
describe the order of the edges of \Gint around $\phi$.  Initialize
the rotation scheme of $\phi$ to $(\omega,\phi)$,$(a_o,\phi)$,
$(y^*_0,\phi)$, $(b_o,\phi)$, and, for $i=1,\dots,m$, $(x^*_i,\phi)$,
$(a_i,\phi)$, $(y^*_i,\phi)$, and $(b_i,\phi)$. Then, for
$i=1,\dots,m$, initialize $first_i=a_i$ and $last_i=x^*_i$. For each
element $j$ of $\mathcal{O}$, place $(x^j_i,\phi)$ between
$(first_i,\phi)$ and $(last_i,\phi)$ in the rotation scheme of $\phi$,
and set $first_i=x_i^j$. Also, for $i=0,\dots,m$, initialize
$first_i=y^*_i$ and $last_i=b_i$. For each element $j$ of
$\mathcal{O}$, place $(y^j_i,\phi)$ between $(first_i,\phi)$ and
$(last_i,\phi)$ in the rotation scheme of $\phi$, and set
$last_i=y_i^j$. Refer to Fig.~\ref{fig:3book} for an illustration of
the construction of the rotation scheme of $\phi$.

The rest of the construction of \GammaR, \GammaB, and \GammaG and the
proof that such embeddings determine a \sefe of \sunsefeinstance works
as in the proof of Lemma~\ref{le:aux-ptkbe}. In particular, the fact
that the rotation scheme of $\phi$ determines a planar embedding of
the triconnected subgraphs of \Gg induced by vertices $\phi$, $a_i$,
$x_i^\alpha$, $x_i^\beta$, $x_i^\gamma$, for $i=1,\dots,m$, derives
from the fact that $\mathcal{O}$ is a solution of instance $\langle
A,C \rangle$ of \betp.  This concludes the proof of the theorem.
\end{proof}

Although \ptckpbepshort has been shown \NPC for $k\geq 3$ even when two of the input graphs are
biconnected in Theorem~\ref{th:pt3be-2bico}, we show that stronger
conditions on the connectivity of the graphs allow for a
polynomial-time solution of the problem. As observed before, the
linear-time algorithm by Hoske~\cite{hoske-befpa-12} for \pkpbepshort
when each graph is $T$-biconnected can be easily extended to solve
\ptckpbepshort under the same conditions. In the following theorem we
prove that for $k\geq 2$ this is true even if only $k-1$ graphs are $T$-biconnected.

At this aim, we describe an algorithm that we call {\sc
  ALGO-$(k-1)$-$T$-BICO} to decide whether an instance $\langle T, E_1,
\dots, E_k \rangle$ of \ptckpbepshort is positive in the case in which
$k-1$ graphs $G_i$ are $T$-biconnected. In the description of the
algorithm we assume, without loss of generality, that graphs
$G_1,\dots,G_{k-1}$ are $T$-biconnected.

\begin{itemize}
\item[{\bf STEP 1.}] For $i=1,\dots,k-1$, we construct an auxiliary
  graph $H_i$ as follows. Initialize $H_i$ to $G_i$; remove from $H_i$
  the internal vertices of $T$ and their incident edges; and add to
  $H_i$ a vertex $w_i$ and connect it to all vertices of $H_i$ (that
  is, to all leaves of $T$).
\item[{\bf STEP 2.}] For $i=1,\dots,k-1$, we construct a PQ-tree
  $\mathcal{T}_i$ representing all possible orders of the edges around
  $w_i$ in a planar embedding of $H_i$ by applying the planarity
  testing algorithm of Booth and
  Lueker~\cite{bt-tcopiggppqa-76}. Since, by construction, all
  vertices of $H_i$ different from $w_i$ are adjacent to $w_i$, the
  leaves of $\mathcal{T}_i$ are in one-to-one correspondence with the
  leaves of $T$. Hence, all PQ-trees $\mathcal{T}_i$ have the same
  leaves.
\item[{\bf STEP 3.}] We intersect all PQ-trees
  $\mathcal{T}_1,\dots,\mathcal{T}_{k-1}$ to obtain a PQ-tree
  $\mathcal{T}^*$ representing all the possible partitioned
  book embeddings of graphs $H_i\setminus w_i$, for
  $i=1,\dots,k-1$. We remark that the procedure described so far is
  analogous to the one described in~\cite{hoske-befpa-12} to compute a
  \pkpbepshort of $k$ $T$-biconnected graphs.
\item[{\bf STEP 4.}] We intersect $\mathcal{T}^*$ with $T$ to obtain a
  PQ-tree $\mathcal{T}$ representing all the possible partitioned
  $T$-coherent book embeddings of instance $\langle T,
  E_1,\dots,E_{k-1}\rangle$.
\item[{\bf STEP 5.}] We construct a {\em representative graph}
  $G_{\mathcal{T}}$ from $\mathcal{T}$, as described
  in~\cite{FengCE95}, composed of {\em wheel} graphs (that is, graphs
  consisting of a {\em central vertex} and of a cycle, called the {\em
    rim} of the wheel, such that the central vertex is connected to
  every vertex of the rim), edges connecting vertices of the rims of
  different wheels not creating simple cycles containing vertices belonging to more than one wheel, and vertices of degree $1$, which are in
  one-to-one correspondence with the leaves of $\mathcal{T}$, each
  connected to a vertex of the rim of some wheel.
\item[{\bf STEP 6.}] We extend graph $G_{\mathcal{T}}$ by adding an
  edge between two degree-$1$ vertices if and only if the two leaves
  of $T$ corresponding to such vertices are connected by an edge of
  $E_k$; hence obtaining graph $H$.
\item[{\bf STEP 7.}] We return \texttt{YES} if  $H$ is planar, otherwise we return
  \texttt{NO}.
\end{itemize}

In the following theorem we prove the correctness and the time complexity of {\sc ALGO-$(k-1)$-$T$-BICO},
\begin{theorem}\label{th:algo}
  Let $\langle T, E_1, \dots, E_k \rangle$ be an instance of
  \ptckpbepshort with $k\geq 2$ in which $k-1$ graphs are
  $T$-biconnected. There exists an $O(k\cdot n)$-time algorithm to
  decide whether $\langle T, E_1, \dots, E_k \rangle$ admits a
  \ptckpbep, where $n$ is the number of vertices of $T$.
\end{theorem}

\begin{proof}
  The algorithm that decides \ptckpbepshort for $\langle T, E_1,
  \dots, E_k \rangle$ is {\sc ALGO-$(k-1)$-$T$-BICO}.

  We prove the correctness. First, observe that, as proved
  in~\cite{hoske-befpa-12}, the PQ-tree $\mathcal{T}^*$ constructed at
  STEP 3 encodes all and only the partitioned $(k-1)$-page
  book embeddings of instance $\langle
  \mathcal{L}(T),E_1,\dots,E_{k-1}\rangle$. Thus, intersecting
  $\mathcal{T}^*$ with tree $T$ yields a PQ-tree $\mathcal{T}$ (see
  STEP 4) encoding all and only the partitioned $T$-coherent
  $(k-1)$-page book embeddings\footnote{This is the extension of the
    algorithm by Hoske to instances of \ptckpbepshort mentioned
    before.} of instance $\langle T, E_1, \dots, E_{k-1} \rangle$.

  Also, as proved in~\cite{FengCE95}, there exists a one-to-one
  correspondence between the possible orderings of the leaves of
  $\mathcal{T}$ and the possible orderings obtained by restricting the
  order of the vertices in an Eulerian tour of the outer face in a
  planar embedding of $G_{\mathcal{T}}$ to the degree-$1$ vertices.

  Given a planar embedding $\Gamma$ of $H$ (see Fig.~\ref{fig:ALGO1}), we construct a partitioned
  $T$-coherent $k$-page book embedding $\mathcal{O}$ of $\langle T,
  E_1, \dots, E_k \rangle$.  We claim that $\Gamma$ can be modified in
  order to obtain a planar embedding $\Gamma'$ of $H$  (see Fig.~\ref{fig:ALGO3}) such that all
  the degree-$1$ vertices of $G_{\mathcal{T}}$ lie on the outer face
  of the embedding $\Gamma_{\mathcal{T}}$ of $G_{\mathcal{T}}$
  obtained by restricting $\Gamma'$ to the vertices and edges of
  $G_{\mathcal{T}}$.

  The claim implies that the order $\mathcal{O}$ of the degree-$1$
  vertices in a Eulerian tour of the outer face of
  $\Gamma_{\mathcal{T}}$ is a partitioned $T$-coherent $k$-page book
  embedding of $\langle T, E_1, \dots, E_k \rangle$ since (i)
  $\mathcal{O}$ is represented by $\mathcal{T}$ and (ii) no two edges
  of $E_k$ alternate in $\mathcal{O}$, given that $\Gamma'$ is planar.

  We prove the claim. First, we show that starting from $\Gamma$ we
  can obtain a planar drawing $\Gamma^*$ of $H$ such that every wheel
  of $G_{\mathcal{T}}$ is drawn {\em canonically}  (see Fig.~\ref{fig:ALGO2}), namely, with its
  central vertex lying in the interior of its rim. Consider any wheel
  $W$ of $G_{\mathcal{T}}$ with central vertex $\omega$ that is not
  drawn canonically in $\Gamma$. This implies that there exist two
  vertices $a$ and $b$ of the rim of $W$ such that all the vertices of
  $W$ different from $a$, $b$, and $\omega$ lie in the interior of
  cycle $\langle a,b,\omega \rangle$. Since, by construction of
  $G_{\mathcal{T}}$ and of $H$, vertex $\omega$ is not adjacent to any
  vertex not belonging to $W$, it is possible to reroute edge $(a,b)$
  as a curve arbitrarily close to path $(a,\omega,b)$ so that cycle
  $\langle a,b,\omega \rangle$ does not enclose any vertex of
  $H$. Observe that, such an operation might determine a change in the
  rotation scheme of $a$ or $b$. Applying such a procedure to all
  non-canonically drawn wheels, eventually results in a planar drawing
  $\Gamma^*$ of $H$ such that all wheels of $G_{\mathcal{T}}$ are
  drawn canonically.  Second, we show how to obtain $\Gamma'$ starting
  from $\Gamma^*$  (see Fig.~\ref{fig:ALGO3}). Consider any wheel $W$ of $G_{\mathcal{T}}$, with
  central vertex $\omega$. For each two adjacent vertices $a$ and $b$
  of the rim of $W$, if there exist vertices of $H$ in the interior of
  cycle $\langle a,b,\omega \rangle$, then we reroute edge $(a,b)$ as
  a curve arbitrarily close to path $(a,\omega,b)$ so that cycle
  $\langle a,b,\omega \rangle$ does not enclose any vertex of
  $H$. Since $\omega$ is not connected to vertices of $H$ other than
  those belonging to the rim of $W$, this operation does not introduce
  any crossing. After this operation has been performed for every two
  adjacent edges of the rim of $W$, there exists no vertex of $H$ not
  belonging to $W$ in the interior of the rim of $W$, since $W$ is
  drawn canonically. This concludes the proof of the claim, since
  $G_{\mathcal{T}}$ does not contain any simple cycle containing
  vertices belonging to more than one wheel and no wheel of
  $G_{\mathcal{T}}$ contains in its interior vertices of $H$ not
  belonging to it.

    \begin{figure}[htb]
      \centering 
      \subfigure[]{
        \includegraphics[width=.3163\textwidth]{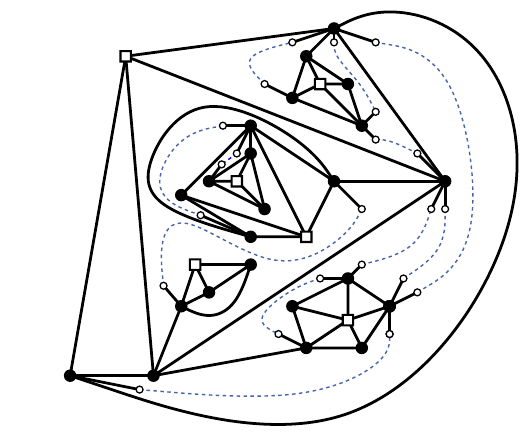}\label{fig:ALGO1}
      }
      \subfigure[]{
        \includegraphics[width=.3163\textwidth]{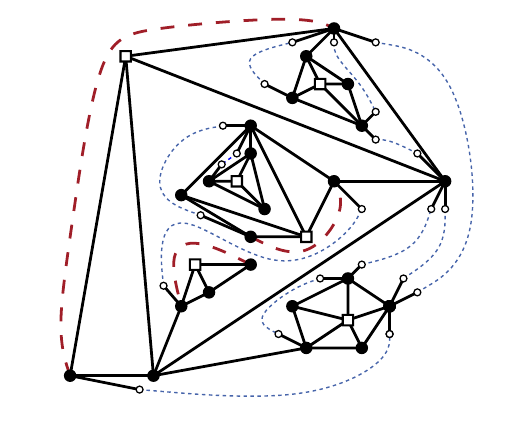}\label{fig:ALGO2}
      }
      \subfigure[]{
        \includegraphics[width=.3163\textwidth]{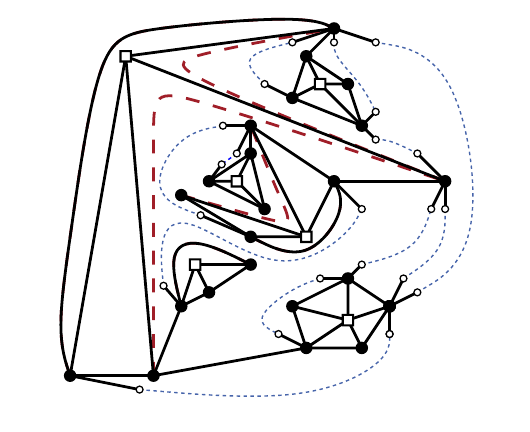}\label{fig:ALGO3}
      }
      \caption{Illustration for the proof of Theorem~\ref{th:algo}. Edges of $G_\mathcal{T}$ are black solid curves. Edges of $E_k$ are blue dotted curves. Edges of $H$ which have been redrawn with respect to the previous drawing are red dashed curves. Central vertices of the wheels are white squares. Degree-$1$ vertices of $G_{\mathcal{T}}$ are white circles. (a) Planar drawing $\Gamma$ of $H$. (b) Planar drawing $\Gamma^*$ of $H$ in which every wheel of  $G_\mathcal{T}$ is drawn canonically. (c) Planar drawing $\Gamma'$ of $H$ in which all the degree-$1$ vertices of  $G_\mathcal{T}$ lie in the outer face of $\Gamma'$ restricted to $G_\mathcal{T}$.}\label{fig:ALGOPROOF}
    \end{figure}

  Given a partitioned $T$-coherent $k$-page book embedding
  $\mathcal{O}$ of $\langle T, E_1, \dots, E_k\rangle$, we construct a
  planar embedding $\Gamma$ of $H$. To obtain $\Gamma$, we first
  augment $G_{\mathcal{T}}$ to an auxiliary graph $U$ by adding a
  dummy edge between two degree-$1$ vertices of $G_{\mathcal{T}}$ if
  and only if the corresponding leaves of $T$ are either adjacent in
  $\mathcal{O}$ or appear as the first and last element in
  $\mathcal{O}$. Since $\mathcal{O}$ is a partitioned $T$-coherent
  $k$-page book embedding $\mathcal{O}$ of $\langle T, E_1, \dots,
  E_k\rangle$, it is possible to find a planar embedding of
  $G_{\mathcal{T}}$ in which the degree-$1$ vertices appear along the
  Eulerian tour of the outer face in the same order as $\mathcal{O}$.
  Hence, graph $U$ is planar.  Produce a planar drawing $\Gamma^*$ of
  $H$ whose outer face is the cycle composed of all the dummy
  edges. Since $\mathcal{O}$ is a partitioned $T$-coherent $k$-page
  book embedding, no two edges of $E_k$ alternate in
  $\mathcal{O}$. Hence they can be drawn in the outer face of
  $\Gamma^*$ without introducing crossings. Removing all dummy edges
  yields a planar embedding $\Gamma$ of $H$.
 
  We prove the time complexity. STEP~1 and STEP~2 take $O(k\cdot n)$
  time, since the time-complexity of constructing a PQ-tree on a
  ground set of $n$ elements is linear in the size of the ground
  set~\cite{b-pqra-75,bt-tcopiggppqa-76}. STEP~3 and STEP~4 take
  $O((k-2)\cdot n)$ and $O(n)$ time, respectively, since the
  intersection of two PQ-trees can be performed in amortized linear
  time in their size~\cite{b-pqra-75} and the size of the obtained
  PQ-tree stays linear in the size of the ground set. STEP~5 takes
  linear time in the size of $\mathcal{T}$, since it corresponds to
  replacing each Q-node with a wheel and each P-node with a cut vertex
  connecting the wheels~\cite{FengCE95}. Observe that, graph
  $G_\mathcal{T}$ has size linear in $n$, since each vertex of the rim
  of a wheel corresponds to exactly one edge of $\mathcal{T}$. STEP~6
  takes $O(|E_k|)=O(n)$ time and produces a graph $H$ with $O(n)$ vertices. Finally, testing the planarity of $H$
  takes linear time in the size of $H$~\cite{bt-tcopiggppqa-76}.
 
  This concludes the proof of the theorem.
\end{proof}

\subsection{\ptcTWOpbep}

In this subsection we restrict our attention to instances
\ptcTWOpbeinstance of \ptckpbepshort with $k=2$. We remark that this
problem has been proved~\cite{adfpr-tsetgibgt-11} equivalent to \sefep
for $k=2$ when the intersection graph \Gint is connected. This problem
was only known to be polynomial-time solvable if (i) $T$ is a
star~\cite{hn-tpbecgp-09}, (ii) $\red{G_1}=(V(T),E(T) \cup \red{E_1})$
and $\blue{G_2}=(V(T),E(T) \cup \blue{E_2})$ are
biconnected~\cite{br-spqacep-13}, or (iii) $T$ is
binary~\cite{hoske-befpa-12,s-ttphtpv-13}. Theorem~\ref{th:algo} extends the class of polynomially-solvable instances by showing that
\ptcTWOpbepshort is linear-time solvable if either \Gr or \Gb is
$T$-biconnected.

In the following we prove that, in order to find a polynomial-time
algorithm for the general setting of \ptcTWOpbepshort, it suffices to focus on instances of
\ptcTWOpbepshort in which only one of the two graphs is biconnected (not
$T$-biconnected) and series-parallel.

\begin{theorem}\label{th:seriesparallel}
  Let \ptcTWOpbeinstance be an instance of \ptcTWOpbepshort. There
  exists an equivalent instance \ptcTWOpbeinstancep{*} of
  \ptcTWOpbepshort such that one of the two graphs is biconnected and
  series-parallel.
\end{theorem}
\begin{proof}
  We describe how to construct instance \ptcTWOpbeinstancep{*}
  starting from \ptcTWOpbeinstance. Refer to
  Fig~\ref{fig:seriesparallel}.

\begin{figure}[htb]
  \centering
  \subfigure[\ptcTWOpbeinstance]{\includegraphics[width=0.45\textwidth]{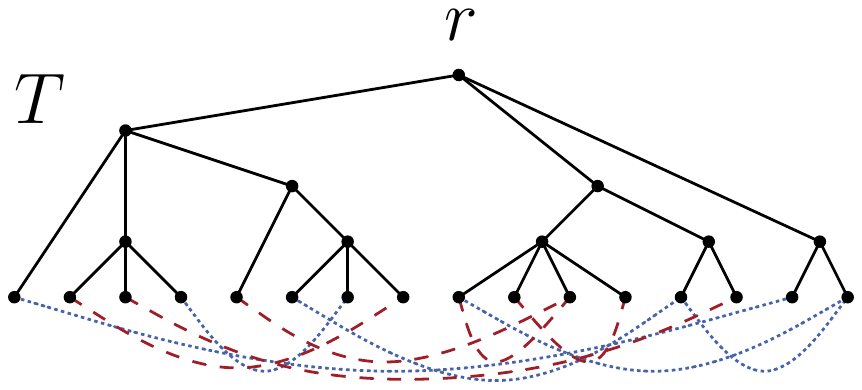}}
  \subfigure[\ptcTWOpbeinstancep{*}]{
    \includegraphics[width=0.45\textwidth]{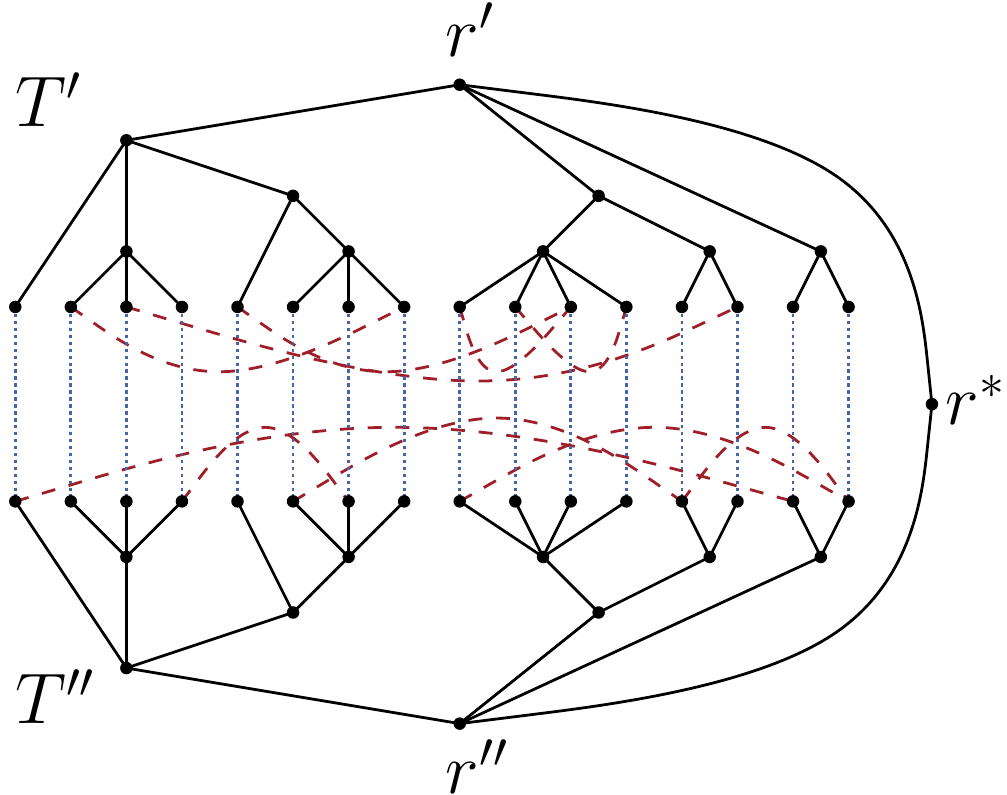}}
  \caption{Illustration of the proof of
    Theorem~\ref{th:seriesparallel}.}\label{fig:seriesparallel}
\end{figure}
Let $r$ be any internal vertex of $T$.  Tree $T^*$ is constructed as
follows. Initialize tree $T^*$ to the union of two copies $T'$ and
$T''$ of $T$. For each vertex $v \in T$, let $v'$ and $v''$ be the two
copies of $v$ in $T'$ and in $T''$, respectively. Add a vertex $r^*$
to $T^*$ and edges $(r^*,r')$ and $(r^*,r'')$.
Sets \Erp{*} and \Ebp{*} are defined as follows. Set
$\red{E_1^*}=\{(v'_i,v'_j): (v_i,v_j) \in \red{E_1}\} \cup
\{(v''_i,v''_j): (v_i,v_j) \in \blue{E_2}\}$. Also, set
$\blue{E_2^*}=\{(v_i',v''_i): v_i \in \mathcal{L}(T)\}$, where
$\mathcal{L}(T)$ denotes the set of leaves of $T$.

It is straightforward to observe that, by construction, the graph
\GbP{*} composed of $\mathcal{T}^*$ plus the edges in $\blue{E_2^*}$
is biconnected and series-parallel. We prove that
\ptcTWOpbeinstancep{*} is equivalent to \ptcTWOpbeinstance.

Suppose that \ptcTWOpbeinstance admits a partitioned $T$-coherent
$2$-page book embedding $\mathcal{O}$.  We construct an order
$\mathcal{O}^*$ for \ptcTWOpbeinstancep{*} as follows.  For each
$i=1,\dots,|\mathcal{L}(T)|$, consider the vertex $v_j$ at position
$i$ in $\mathcal{O}$. Place vertices $v_j'$ and $v_j''$ at positions
$i$ and $2\cdot |\mathcal{L}(T)|-i+1$ in $\mathcal{O}^*$,
respectively.

We prove that $\mathcal{O}^*$ is a partitioned $T$-coherent $2$-page
book embedding of \ptcTWOpbeinstancep{*}. First, we observe that
$\mathcal{O}^*$ is represented by $T^*$, as (i) $T^*$ is composed of
two copies of $T$ connected through $r^*$, (ii) $\mathcal{O}^*$ is
composed of two suborders of which the first coincides with
$\mathcal{O}$ and the second coincides with the reverse of
$\mathcal{O}$, where each element $v_j$ of $\mathcal{O}$ is identified
with elements $v'_j$ and $v''_j$ of $\mathcal{O}^*$, and (iii)
$\mathcal{O}$ is represented by $T$.  Second, we prove that the
endvertices of edges in \Erp{*} and \Ebp{*} do not alternate in
$\mathcal{O}^*$. As for the edges in \Ebp{*}, we observe that for
every two edges $(v'_i,v_i'')$ and $(v'_j,v_j'')$ with $i<j$, both
vertices $v'_j$ and $v_j''$ lie between $v'_i$ and $v_i''$ in
$\mathcal{O}^*$.  As for the edges in \Erp{*}, we first observe that
no alternation occurs between the endvertices of edges $(v'_i,v'_j)$
and $(v''_h,v''_k)$ as both $v'_i$ and $v'_j$ appear in
$\mathcal{O}^*$ before $v''_h$ and $v''_k$, by construction. Also, no
two edges $(v'_i,v_j')$ and $(v'_h,v'_k)$ alternate in $\mathcal{O}^*$
as otherwise edges $(v_i,v_j)$ and $(v_h,v_k)$ would alternate in
$\mathcal{O}$.  For the same reason, no two edges $(v''_i,v''_j)$ and
$(v''_h,v''_k)$ alternate in $\mathcal{O}^*$.

Suppose that \ptcTWOpbeinstancep{*} admits a partitioned $T$-coherent
$2$-page book embedding $\mathcal{O}^*$.  We first observe that in
$\mathcal{O}^*$ either all vertices $v'_i\in T'$ appear consecutively
or all vertices $v''_i\in T''$ do, as $\mathcal{O}^*$ is represented
by $T^*$ and $T^*$ consists of the two copies $T'$ and $T''$ of $T$.
Also, given a partitioned $T$-coherent $2$-page book embedding
$\mathcal{O}^1$, it is possible to obtain a new one $\mathcal{O}^2$ by
performing a circular shift on the elements of $\mathcal{O}^1$, that
is, by setting the first element of $\mathcal{O}^1$ as the last
element of $\mathcal{O}^2$ and by setting the element at position $i$
in $\mathcal{O}^1$ as the element at position $i-1$ in
$\mathcal{O}^2$, for each $i=2,\dots,|\mathcal{O}^1|$. Hence, in the
following, we will assume that $\mathcal{O}^*$ is such that all the
vertices $v'_i\in T'$ appear before all the vertices $v''_j \in T''$.

We construct an order $\mathcal{O}$ for \ptcTWOpbeinstance as follows.
For each $i=1,\dots,|\mathcal{L}(T')|$, consider the vertex $v'_j$ at
position $i$ in $\mathcal{O}^*$ and place vertex $v_j$ at position $i$
in $\mathcal{O}$.

We prove that $\mathcal{O}$ is a partitioned $T$-coherent $2$-page
book embedding of \ptcTWOpbeinstance. First, we observe that
$\mathcal{O}$ is represented by $T$, as the suborder of
$\mathcal{O}^*$ restricted to its first $|\mathcal{L}(T)|$ elements
(that corresponds to a copy of $\mathcal{O}$) is represented by $T'$
(that is a copy of $T$, where vertex $v'_i \in T'$ is identified with
vertex $v_i \in T$).  Second, we prove that the endvertices of edges
in \Er and \Eb do not alternate in $\mathcal{O}$. In order to prove
that, first observe that the suborder $\mathcal{O}'$ of
$\mathcal{O}^*$ restricted to its first $|\mathcal{L}(T)|$ elements is
the reverse of the suborder $\mathcal{O}''$ of $\mathcal{O}^*$
restricted to its last $|\mathcal{L}(T)|$ elements, where vertex $v'_i
\in T'$ is identified with vertex $v''_i \in T''$. This is due to the
fact that (i) for every $i=1,\dots,|\mathcal{L}(T)|$, there exists
edge $(v'_i,v''_i)$ and (ii) all the vertices $v'_i \in T'$ appear
before all the vertices $v''_j \in T''$. This implies that if the
endvertices of two edges $(v_i,v_j)$ and $(v_h,v_k)$ belonging to \Er
(to \Eb) alternate in $\mathcal{O}$, then the corresponding copies
$v'_i$, $v'_j$, $v'_h$, and $v'_k$ (the corresponding copies $v''_i$,
$v''_j$, $v''_h$, and $v''_k$) alternate in $\mathcal{O}^*$. However,
this contradicts the fact that $\mathcal{O}^*$ is a partitioned
$T$-coherent $2$-page book embedding of \ptcTWOpbeinstance, since
edges $(v'_i,v'_j)$ and $(v'_h,v'_k)$ (edges $(v''_i,v''_j)$ and
$(v''_h,v''_k)$) exist in \Erp{*} by construction.  This concludes the
proof of the theorem.
\end{proof}

\section{\maxsefep}\label{se:maxsefe}

In this section we study the optimization version of the \sefe
problem, in which two embeddings of the input graphs \Gr and \Gb are
searched so that as many edges of \Gint as possible are drawn the
same. We study the problem in its decision version and call it
\maxsefep. Namely, given a triple \maxsefeinstance{k^*} composed of
two planar graphs \Gr and \Gb, and an integer $k^*$, the \maxsefep
problem asks whether \Gr and \Gb admit a simultaneous embedding
\sefesolution in which at most $k^*$ edges of \Gint have a different
drawing in \GammaR and in \GammaB.  First, in
Lemma~\ref{le:MaxSefeNP}, we state the membership of \maxsefep to \NP,
which descends from the fact that \sefe belongs to \NP. Then, in
Theorem~\ref{th:fixedembedding} we prove the \NPCN in the general
case. Finally, in Theorem~\ref{th:degree-two}, we prove that the
problem remains \NPC even if stronger restrictions are imposed on the
intersection graph \Gint of \Gr and \Gb.

\begin{lemma}\label{le:MaxSefeNP}
  \maxsefep is in \NP.
\end{lemma}
\remove{
  \begin{proofsketch}
    One can guess any set of at most $k^*$ edges not to be drawn the
    same, subdivide each of them in one of the two graphs, and test
    for the existence of a \sefe.  Since \sefe belongs to
    \NP~\cite{gjpss-sgefe-06}, the statement follows.
  \end{proofsketch}
}

\begin{proof}
  The statement descends from the fact that the \sefe problem belongs
  to \NP~\cite{gjpss-sgefe-06}. Namely, let \maxsefeinstance{k^*} be
  an instance of \maxsefep.  Non-deterministically construct in
  polynomial time all the sets of at most $k^*$ edges of \Gint. Then,
  for each of the constructed sets, replace every edge in the set with
  a path of length $2$ in one of the two graphs, say \Gr, hence
  obtaining a graph \textcolor{red}{$G'_1$}, and test whether a \sefe
  of \textcolor{red}{$G'_1$} and \Gb exists in polynomial time with a
  non-deterministic Turing machine~\cite{gjpss-sgefe-06}. If at least
  one of the performed tests succeeds, then \maxsefeinstance{k^*} is a
  positive instance.
\end{proof}

In order to prove that \maxsefep is \NPC, we show a reduction from a
variant of the \NPC problem \stplong (\stp)~\cite{gj-rstpnpc-77},
defined as follows: Given an instance $\langle G(V,E), S, k \rangle$
of \stp, where $G(V,E)$ is a planar graph whose edges have weights
$\omega:E\rightarrow \mathbb{N}$, $S\subset V$ is a set of
\emph{terminals}, and $k >0$ is an integer; does a tree $T^*(V^*,E^*)$
exist such that (1) $V^*\subseteq V$, (2) $E^*\subseteq E$, (3)
$S\subseteq V^*$, and (4) $\sum_{e\in E^*}\omega(e)\leq k$? The edge
weights in $\omega$ are bounded by a polynomial function $p(n)$
(see~\cite{gj-rstpnpc-77}).
In our variant, that we call \upstp (\upstshort), graph $G$ is a
triconnected planar graph and all the edge weights are equal to $1$.
We remark that a variant of \stp in which all the edge weights are
equal to $1$ and in which $G$ is a \emph{subdivision} of a
triconnected planar graph (and no subdivision vertex is a terminal) is
known to be \NP-complete~\cite{addfr-bcp-12}. However, using this
variant of the problemwould create multiple edges in our
reduction. Actually, the presence of multiple edges might be handled
by replacing them in the constructed instance with a set of length-$2$
paths. However, we think that an \NPCN proof for the \stp problem with
$G$ triconnected and uniform edge weights may be of independent
interest.

\begin{lemma}\label{le:steiner-triconnected-uniform-hard}
  \upstp is \NPC.
\end{lemma}

\remove{
  \begin{proofsketch}
    Since instance of \upstshort is also an instance of \stp, we have
    that \upstshort belongs to \NP.
    \begin{figure}[htb]
      \centering \subfigure[]{
        \includegraphics[width=.226\textwidth]{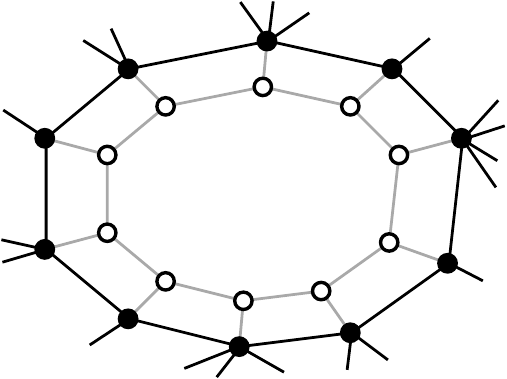}\label{fig:webgadget}
      }\hspace{2cm} \subfigure[]{
        \includegraphics[width=.226\textwidth]{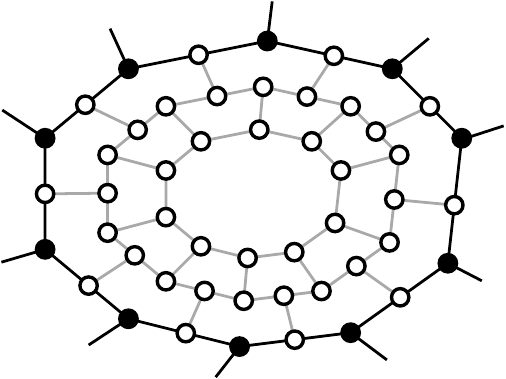}\label{fig:degree3}
      }
      \caption{(a) Gadget added inside a face to make $G^*$
        triconnected. (b) Gadget replacing a vertex of degree greater
        than $3$ to make \Gint subcubic.}
    \end{figure}
    The \NPHN is proved by means of a polynomial-time reduction from
    \stp.  Namely, given an instance $\langle G, S, k \rangle$ of
    \stp, we construct an equivalent instance $\langle G',S',k'
    \rangle$ of \upstshort as follows. First, augment $G$ to a
    triconnected planar graph $G'$ by adding dummy edges whose weight
    is the sum of the weight of the edges of $G$, that is bounded by
    function $p(n)$. Then, in order to obtain uniform weights, replace
    each edge $e$ in $G'$ with a path $P(e)$ of $\omega(e)$ weight-$1$
    edges. Finally, since the subdivision vertices have degree $2$,
    add inside each face of the unique planar embedding of $G'$ a
    gadget as the one in Fig. \ref{fig:webgadget}, all of whose edges
    have weight $1$.
  \end{proofsketch}
}

\begin{proof}
  The membership in \NP follows from the fact that an instance of
  \upstshort is also an instance of \stp.
  \begin{figure}[htb]
    \centering \subfigure[]{
      \includegraphics[width=.226\textwidth]{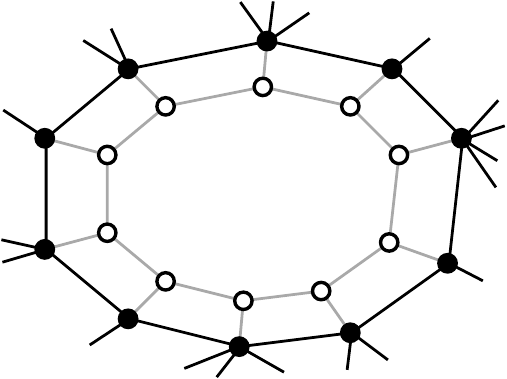}\label{fig:webgadget}
    }\hspace{2cm} \subfigure[]{
      \includegraphics[width=.226\textwidth]{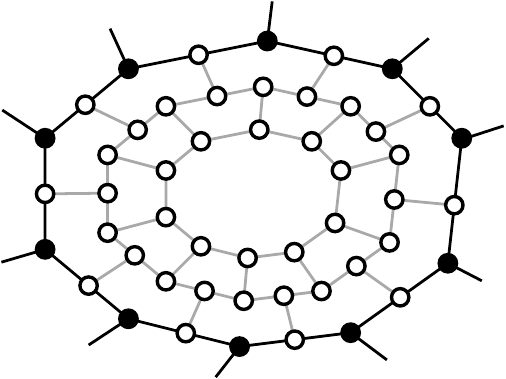}\label{fig:degree3}
    }
    \caption{(a) Gadget added inside a face to make $G^*$
      triconnected. (b) Gadget replacing a vertex of degree greater
      than $3$ to make \Gint subcubic.}
  \end{figure}
  The \NPHN is proved by means of a polynomial-time reduction from
  \stp.  Let $\langle G, S, k \rangle$ be any instance of \stp.  We
  construct an equivalent instance $\langle G', S', k' \rangle$ of
  \upstshort as follows. Initialize $G'=G$.  Let $w = \sum_{e \in G'}
  w(e)$. Since the weights in $\omega$ are bounded by a polynomial
  function $p(n)$, the value of $w$ is also bounded by a polynomial
  function $n\cdot p(n)$.  Augment $G'$ to a triconnected planar graph
  by adding dummy edges and set $\omega(e_d)=w$ for each dummy edge
  $e_d$.  Then, replace each edge $e$ in $G'$ with a path $P(e)$ of
  $\omega(e)$ weight-$1$ edges.  Further, for each face $f$ of the
  unique planar embedding of $G'$, consider the vertices
  $v_1,\dots,v_h$ of $f$ as they appear on the boundary of $f$. Add to
  $G'$ a set $V_f$ of $h$ vertices $u_1,\dots,u_h$ and, for
  $i=1,\dots,h$, add to $G'$ a weight-$1$ edge $(u_i,v_i)$ and a
  weight-$1$ edge $(u_i,u_{i+1})$, where $h+1=1$ (see
  Fig.~\ref{fig:webgadget}). Note that, $G'$ is triconnected.
  Finally, set $S'=S$ and $k'=k$.  Since $w$ is bounded by a
  polynomial function, $\langle G', S', k' \rangle$ can be constructed
  in polynomial time.

  We prove that $\langle G,S,k \rangle$ is a positive instance of \stp
  if and only if $\langle G', S', k' \rangle$ is a positive instance
  of \upstshort.

  Suppose that $\langle G,S,k \rangle$ is a positive instance of
  \stp. Starting from the solution $T$ of $\langle G,S,k \rangle$, we
  construct a solution $T'$ of $\langle G',S',k' \rangle$ by replacing
  each edge $e$ of $T$ with path $P(e)$. By construction, $T'$ is a
  tree, each terminal vertex in $S'$ belongs to $T'$, and $\sum_{e\in
    T'}1 = \sum_{e\in T}\omega(e)\leq k = k'$.

  Suppose that $\langle G',S',k' \rangle$ is a positive instance of
  \upstshort.  Let $T'$ be the solution of $\langle G',S',k'
  \rangle$. Assume that $T'$ is the \emph{optimal} solution of
  $\langle G',S',k' \rangle$, i.e., there exists no solution
  $T^\sharp$ of $\langle G',S',k' \rangle$ such that $\sum_{e\in
    T^\sharp}\omega(e) < \sum_{e\in T'}\omega(e)$.  Observe that, if
  an edge of a path $P(e)$ belongs to $T'$, then all the edges of
  $P(e)$ belong to $T'$, as the internal vertices of $P(e)$ do not
  belong to $S'$, by construction.  Moreover, no edge of a path
  $P(e_d)$ such that $e_d$ is a dummy edge belongs to $T'$, since the
  total weight of the edges of $P(e_d)$ is $w$.  Finally, no edge
  incident to a vertex $u_i \in V_f$, for some face $f$, belongs to
  $T'$, as $S' \cap V_f = \emptyset$ and every path
  $v_i,u_i,\dots,u_l,\dots,u_j,v_j$ connecting two vertices $v_i$ and
  $v_j$ of $f$ and passing through vertices of $V_f$ is two units
  longer than path $v_i,\dots,v_l,\dots,v_j$ only passing through
  vertices of $f$.  Hence, we construct a solution $T$ of $\langle
  G,S,k \rangle$ by replacing in $T'$ all the edges of each path
  $P(e)$ with an edge $e$.  By construction, $T$ is a tree, each
  terminal vertex in $S$ belongs to $T$, and $\sum_{e\in
    T}\omega(e)=\sum_{e\in T'}1 \leq k'=k$. This concludes the proof
  of the lemma.
\end{proof}

Then, based on the previous lemma, we prove the main result of this
section.

\begin{theorem}\label{th:fixedembedding}
  \maxsefep is \NPC.
\end{theorem}

\begin{proof}
  The membership in \NP follows from Lemma~\ref{le:MaxSefeNP}.

  The \NPHN is proved by means of a polynomial-time reduction from
  problem \upstshort.
  Let $\langle G,S,k \rangle$ be an instance of \upstshort. We
  construct an instance \maxsefeinstance{k^*} of \maxsefep as follows
  (refer to Fig.~\ref{fig:fixedembedding}).
  \begin{figure}[htb]
    \centering\ \subfigure[]{
      \includegraphics[width=.4\textwidth]{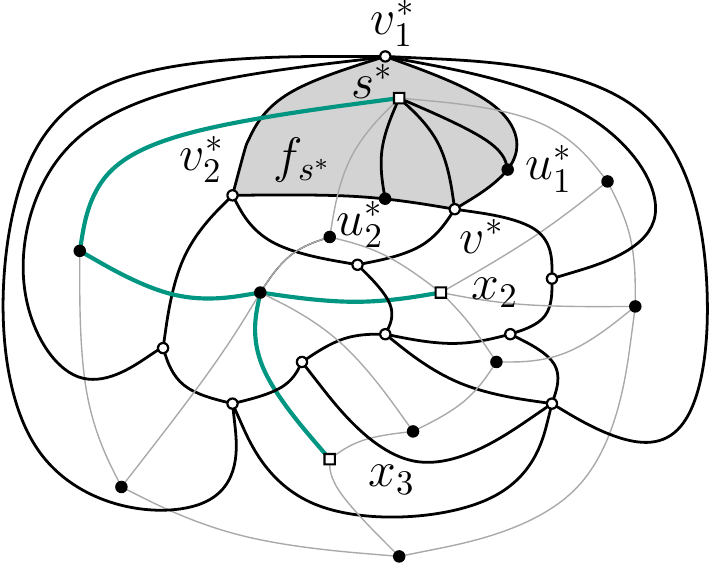}\label{fi:GplusGint}
    } \subfigure[]{
      \includegraphics[width=.4\textwidth]{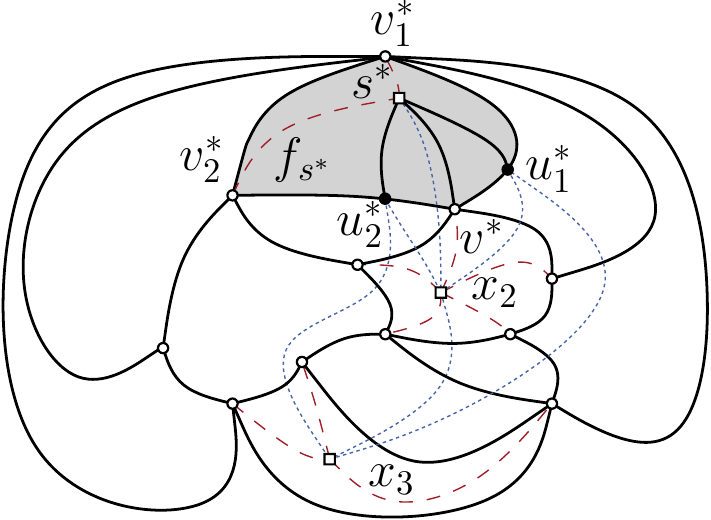}\label{fi:GintPlusGr}
    }
    \\
    \subfigure[]{
      \includegraphics[width=.4\textwidth]{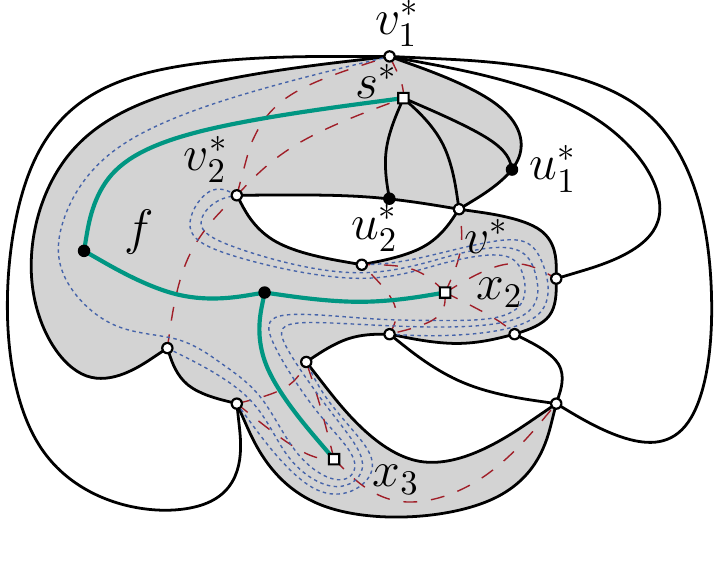}\label{fi:GintRerouted}
    }\subfigure[]{
      \includegraphics[width=.4\textwidth]{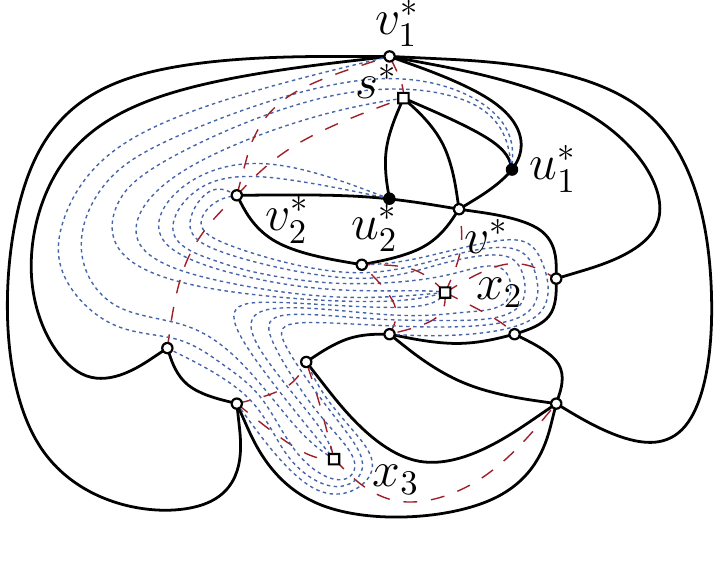}\label{fi:maxsefe}
    }
    \caption{\small Illustration for the proof of
      Theorem~\ref{th:fixedembedding}.  Black lines are edges of
      \Gint; grey lines are edges of $G$; dashed red and solid blue
      lines are edges of \Gr and \Gb, respectively; green edges
      compose the Steiner tree $T$; white squares and white circles
      are terminal vertices and non-terminal vertices of $G$,
      respectively.  (a) \Gint, $G$ and $T$; (b) \Gr$\cup$ \Gb; (c) a
      drawing of \Gint where $4$ edges have two different drawings;
      and (d) a solution \sefesolution of \maxsefeinstance{4}.  }
    \label{fig:fixedembedding}
  \end{figure}

  Since $G$ is triconnected, it admits a unique planar embedding
  $\Gamma_G$, up to a flip. We now construct \Gint, \Gr, and
  \Gb. Initialize \Gint$=$\Gr$\cap$ \Gb as the dual of $G$ with
  respect to $\Gamma_G$. Since $G$ is triconnected, its dual is
  triconnected. Consider a terminal vertex $s^* \in S$, the set
  $E_G(s^*)$ of the edges incident to $s^*$ in $G$, and the face
  $f_{s^*}$ of \Gint composed of the edges that are dual to the edges
  in $E_G(s^*)$. Let $v^*$ be any vertex incident to $f_{s^*}$, and
  let $v^*_1$ and $v^*_2$ be the neighbors of $v^*$ on
  $f_{s^*}$. Subdivide edges $(v^*,v^*_1)$ and $(v^*,v^*_2)$ with
  dummy vertices $u^*_1$ and $u^*_2$, respectively. Add to \Gint
  vertex $s^*$ and edges $(s^*,u^*_1)$, $(s^*,u^*_2)$, and
  $(s^*,v^*)$. Since $v^*$ has at least a neighbor not incident to
  $f_{s^*}$, vertices $u^*_1$ and $u^*_2$ do not create a separation
  pair. Hence, \Gint remains triconnected.  See
  Fig.~\ref{fi:GplusGint}.

  Graph \Gr contains all the vertices and edges of \Gint plus a set of
  vertices and edges defined as follows. For each terminal $s\in S$,
  consider the set $E_G(s)$ of edges incident to $s$ in $G$ and the
  face $f_s$ of \Gint composed of the edges dual to the edges in
  $E_G(s)$. Add to \Gr vertex $s$ and an edge $(s,v_i)$ for each
  vertex $v_i$ incident to $f_s$, without introducing multiple
  edges. Note that, graph \Gr is triconnected. Hence, the rotation
  scheme of each vertex is the one induced by the unique planar
  embedding of \Gr. See Fig.~\ref{fi:GintPlusGr}.

  Graph \Gb contains all the vertices and edges of \Gint plus a set of
  vertices and edges defined as follows. Rename the terminal vertices
  in $S$ as $x_1,\dots,x_{|S|}$, in such a way that $s^*=x_1$. For
  $i=1,\dots,|S|-1$, add edge $(x_i,x_{i+1})$ to \Gb. The rotation
  scheme of the vertices of \Gb different from $x_1,\dots,x_{|S|}$ is
  induced by the embedding of \Gint. The rotation scheme of vertices
  $x_2,\dots,x_{|S|}$ is unique, as they have degree less or equal to
  $2$. Finally, the rotation scheme of $s^*$ is obtained by extending
  the rotation scheme induced by the planar embedding of \Gint, in
  such a way that edges $(s^*,v^*)$ and $(s^*,x_2)$ are not
  consecutive. In order to obtain an instance of \maxsefep in which
  both graphs are triconnected, we can augment \Gb to triconnected by
  only adding edges among vertices
  $\{u^*_1,u_2^*\}\cup\{x_1,\dots,x_{|S|}\}$. See
  Fig.~\ref{fi:GintPlusGr}.
  Finally, set $k^*=k$.

  We show that \maxsefeinstance{k^*} admits a solution if and only if
  $\langle G,S,k \rangle$ does.

  Suppose that $\langle G,S,k \rangle$ admits a solution
  $T$. Construct a planar drawing \GammaR of \Gr. The drawing \GammaB
  of \Gb is constructed as follows. The edges of \Gint that are not
  dual to edges of $T$ are drawn in \GammaB with the same curve as in
  \GammaR. Observe that, in the current drawing \GammaB all the
  terminal vertices in $S$ lie inside the same face $f$ (see
  Fig.~\ref{fi:GintRerouted}).  Hence, all the remaining edges of \Gb
  can be drawn~\cite{pw-epgfvl-01} inside $f$ without intersections,
  as the subgraph of \Gb induced by the vertices incident to $f$ and
  by the vertices of $S$ is planar (see Fig.~\ref{fi:maxsefe}).  Since
  the only edges of \Gint that have a different drawing in \GammaR and
  \GammaB are those that are dual to edges of $T$, \sefesolution is a
  solution for \maxsefeinstance{k^*}.

  Suppose that \maxsefeinstance{k^*} admits a solution \sefesolution
  and assume that \sefesolution is optimal (that is, there exists no
  solution with fewer edges of \Gint not drawn the same). Consider the
  graph $T$ composed of the dual edges of the edges of \Gint that are
  not drawn the same. We claim that $T$ has at least one edge incident
  to each terminal in $S$ and that $T$ is connected. The claim implies
  that $T$ is a solution to the instance $\langle G,S,k \rangle$ of
  \upstshort, since $T$ has at most $k$ edges and since \sefesolution
  is optimal.

  Suppose for a contradiction that there exist two connected
  components $T_1$ and $T_2$ of $T$ (possibly composed of a single
  vertex). Consider the edges of $G$ incident to vertices of $T_1$ and
  not belonging to $T_1$, and consider the face $f_1$ composed of
  their dual edges.  Note that, $f_1$ is a cycle of \Gint. By
  definition of $T$, all the edges incident to $f_1$ have the same
  drawing in \GammaR and in \GammaB. Finally, there exists at least
  one vertex of $S$ that lies inside $f_1$ and at least one that lies
  outside $f_1$. Since all the vertices in $S$ belong to a connected
  subgraph of \Gb not containing any vertex incident to $f_1$, there
  exist two terminal vertices $s'$ and $s''$ such that $s'$ lies
  inside $f_1$, $s''$ lies outside $f_1$, and edge $(s',s'')$ belongs
  to \Gb. This implies that $(s',s'')$ crosses an edge incident to
  $f_1$ in \GammaB, a contradiction.  This concludes the proof of the
  theorem.
\end{proof}


We note from Theorem~\ref{th:fixedembedding} that \maxsefep is \NPC
even if the two input graphs \Gr and \Gb are triconnected, and if the
intersection graph \Gint is composed of a triconnected component and
of a set of isolated vertices (those corresponding to terminal
vertices). We remark that, under these conditions, the original \sefe
problem is polynomial-time solvable (actually, it is polynomial-time
solvable even if only one of the input graphs has a unique
embedding~\cite{adfjkpr-tppeg-10}).
Further, it is possible to transform the constructed instances so that
all the vertices of \Gint have degree at most $3$, by replacing each
vertex $v$ of degree $d(v)>3$ in \Gint with a gadget as in
Fig.~\ref{fig:degree3}. Such a gadget is composed of a cycle of
$2d(v)$ vertices and of an internal grid with degree-$3$ vertices
whose size depends on $d(v)$. Edges incident to $v$ are assigned to
non-adjacent vertices of the cycle, in the order defined by the
rotation scheme of $v$.
Hence, the \maxsefep problem remains \NPC even for instances in which
\Gint is subcubic, that is another sufficient condition to make \sefe
polynomial-time solvable~\cite{s-ttphtpv-13}.

In the following we go farther in this direction and prove that
\maxsefep remains \NPC even if the degree of the vertices in \Gint is
at most $2$. The proof is based on a reduction from the \NPC problem
\xorsatp~\cite{cmsm-tnc-11}, which takes as input (i) a set of Boolean
variables $B=\{x_1,...,x_l\}$, (ii) a $2$-XorSat formula
$F=\bigwedge_{x_i,x_j\in B} (l_i \oplus l_j)$, where $l_i$ is either
$x_i$ or $\overline{x_i}$ and $l_j$ is either $x_j$ or
$\overline{x_j}$, and (iii) an integer $k>0$, and asks whether there
exists a truth assignment $A$ for the variables in $B$ such that at
most $k$ of the clauses in $F$ are not satisfied by $A$.

\begin{theorem}\label{th:degree-two}
  \maxsefep is \NPC even if the intersection graph \Gint of the two
  input graphs \Gr and \Gb is composed of a set of cycles of length
  $3$.
\end{theorem}

\begin{proof}
  The membership in \NP follows from Lemma~\ref{le:MaxSefeNP}.

  The \NPHN is proved by means of a polynomial-time reduction from
  problem \xorsatp.
%
%
  Let $\langle B,F,k \rangle$ be an instance of \xorsatp. We construct
  an instance \maxsefeinstance{k^*} of \maxsefep as follows. Refer to
  Fig.~\ref{fig:xor2}.

\begin{figure}[htb]
  \centering \subfigure[]{
    \includegraphics[height=0.37\textwidth]{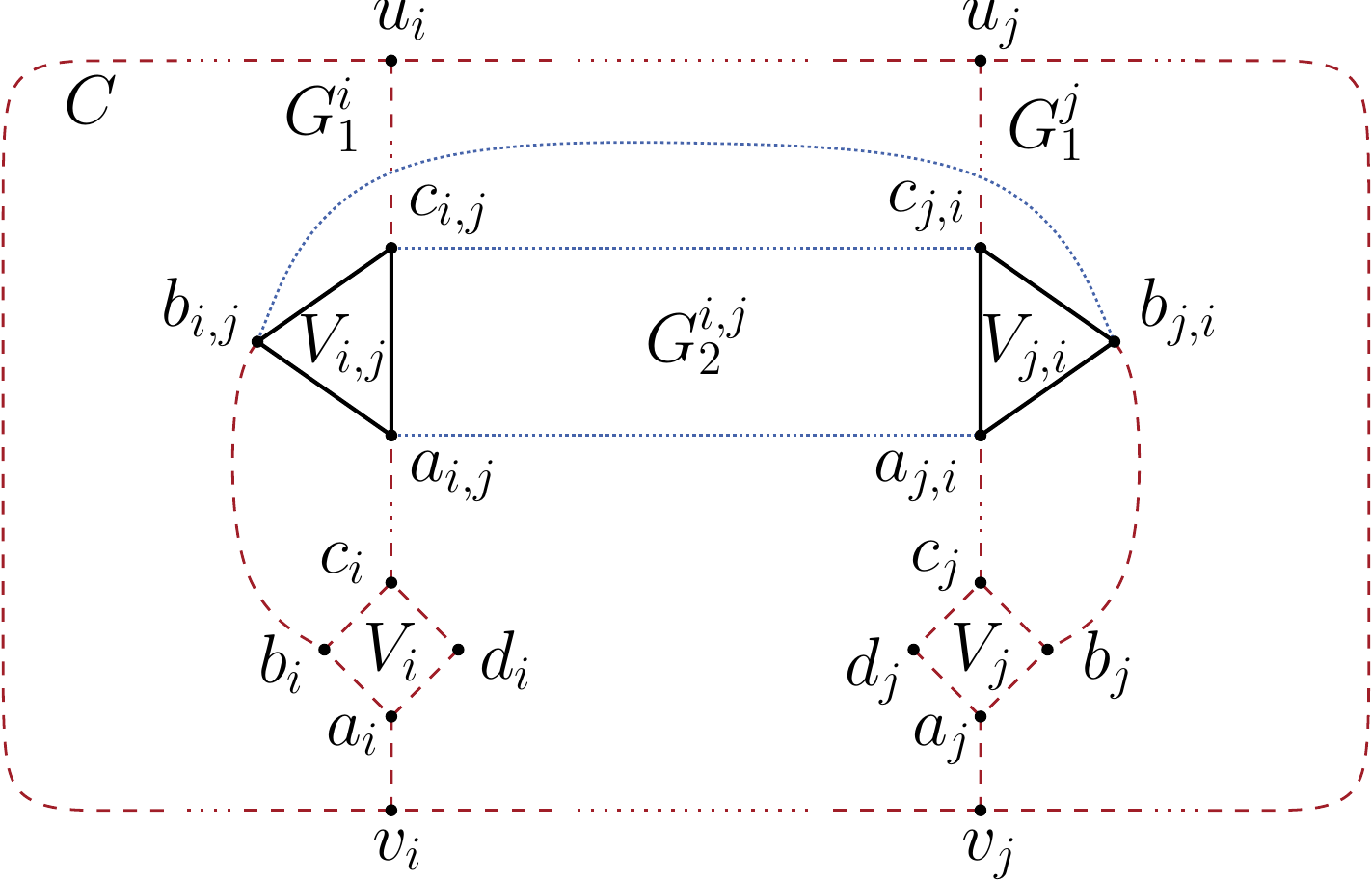}
    \label{fig:xor2}
  } \subfigure[]{
    \includegraphics[height=0.37\textwidth]{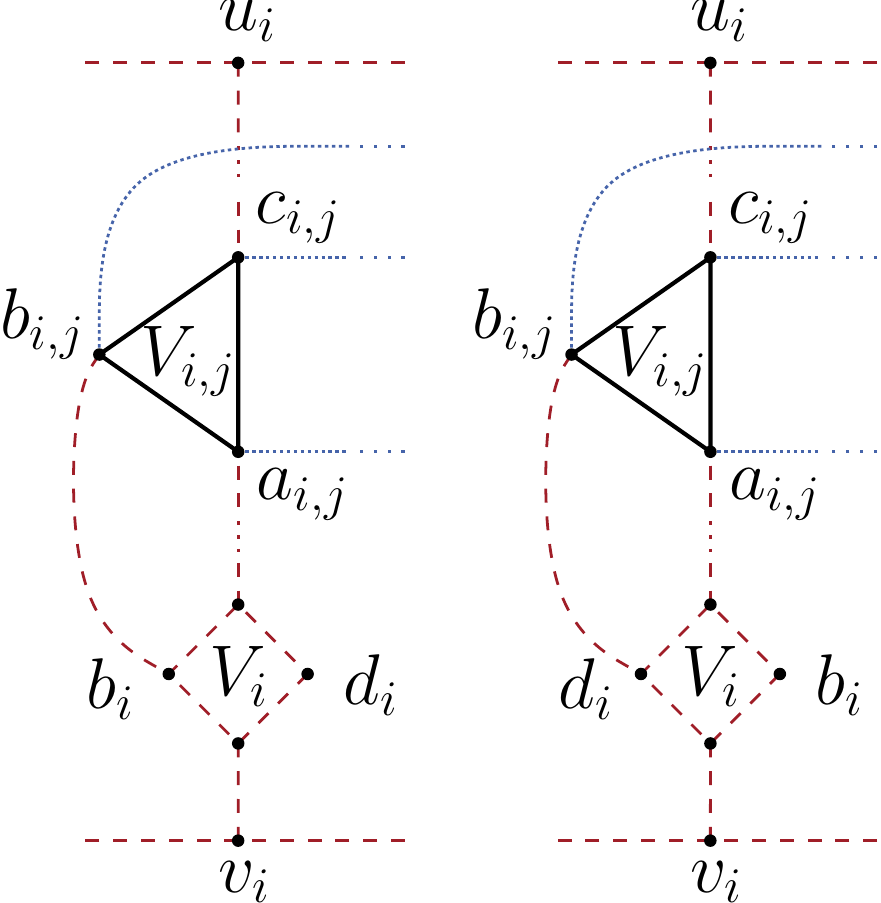}\label{fig:flips}
  }
  \caption{(a) Illustration of the construction of instance
    \maxsefeinstance{k^*} of \maxsefep. (b) Illustration of the two
    cases in which $l_i$ evaluates to \texttt{true} in $A$.}
\end{figure}

Graph \Gr is composed of a cycle $C$ with $2l$ vertices
$v_1,v_2,\dots,v_l,u_l, u_{l-1}, \dots, u_1$. Also, for each variable
$x_i \in B$, with $i=1,\dots,l$, \Gr contains a set of vertices and
edges defined as follows. First, \Gr contains a $4$-cycle
$V_i=(a_i,b_i,c_i,d_i)$, that we call \emph{variable gadget},
connected to $C$ through edge $(a_i,v_i)$. Further, for each clause
$(l_i \oplus l_j) \in F$ (or $(l_j \oplus l_i) \in F$) such that $l_i
\in \{x_i,\overline{x_i}\}$, \Gr contains (i) a $3$-cycle
$V_{i,j}=(a_{i,j},b_{i,j},c_{i,j})$, that we call
\emph{clause-variable gadget}, (ii) an edge $(b_{i,j},w)$, where
either $w=b_i$, if $l_i=x_i$, or $w=d_i$, if $l_i=\overline{x_i}$, and
(iii) an edge $(a_{i,j},c_{i,h})$, where $(l_i \oplus l_h)$ (or $(l_h
\oplus l_i)$) is the last considered clause to which $l_i$
participates; if $(l_i \oplus l_j)$ (or $(l_j \oplus l_i)$) is the
first considered clause containing $l_i$, then $c_{i,h}=c_i$. When the
last clause $(l_i \oplus l_q)$ (or $(l_q \oplus l_i)$) has been
considered, an edge $(c_{i,q},u_i)$ is added to \Gr. Note that, the
subgraph $G^i_1$ of \Gr induced by the vertices of the variable gadget
$V_i$ and of all the clause-variable gadgets $V_{i,j}$ to which $l_i$
participates would result in a subdivision of a triconnected planar
graph by adding edge $(c_{i,q},a_i)$, and hence it has a unique planar
embedding (up to a flip).
Graph \Gb is composed as follows. For each clause $(l_i \oplus l_j)
\in F$, with $l_i \in \{x_i,\overline{x_i}\}$ and $l_j \in
\{x_j,\overline{x_j}\}$, graph \Gb contains a triconnected graph
$G_2^{i,j}$, that we call \emph{clause gadget}, composed of all the
vertices and edges of the clause-variable gadgets $V_{i,j}$ and
$V_{j,i}$, plus three edges $(a_{i,j},a_{j,i})$, $(b_{i,j},b_{j,i})$,
and $(c_{i,j},c_{j,i})$. Finally, set $k^*=k$.

Note that, with this construction, graph \Gint is composed of a set of
$2 |F|$ cycles of length $3$, namely the two clause-variable gadgets
$V_{i,j}$ and $V_{j,i}$ for each clause $(l_i \oplus l_j)$.

We show that \maxsefeinstance{k^*} admits a solution if and only if
$\langle B,F,k \rangle$ does.


Suppose that $\langle B,F,k \rangle$ admits a solution, that is, an
assignment $A$ of truth values for the variables of $B$ not satisfying
at most $k$ clauses of $F$. We construct a solution \sefesolution of
\maxsefeinstance{k^*}.
First, we construct \GammaR. Let the face composed only of the edges
of $C$ be the outer face. For each variable $x_i$, with $i=1,\dots,l$,
if $x_i$ is \texttt{true} in $A$, then the rotation scheme of $a_i$ in
\GammaR is $(a_i,v_i)$, $(a_i,b_i)$, $(a_i,d_i)$ (as in
Fig.~\ref{fig:xor2}).  Otherwise, $x_i$ is \texttt{false} in $A$, and
the rotation scheme of $a_i$ is the reverse (as for $a_j$ in
Fig.~\ref{fig:xor2}). Since $G^i_1$ has a unique planar embedding, the
rotation scheme of all its vertices is univocally determined.
Second, we construct \GammaB. Consider each clause $(l_i \oplus l_j)
\in F$, with $l_i \in \{x_i,\overline{x_i}\}$ and $l_j \in
\{x_j,\overline{x_j}\}$. If $l_i$ evaluates to \texttt{true} in $A$,
then the embedding of $G_2^{i,j}$ is such that the rotation scheme of
$a_{i,j}$ in \GammaB is $(a_{i,j}, b_{i,j})$, $(a_{i,j}, c_{i,j})$,
$(a_{i,j}, a_{j,i})$ (as in Fig.~\ref{fig:xor2}).  Otherwise, $l_i$ is
\texttt{false} in $A$ and the rotation scheme of $a_{i,j}$ is the
reverse (as for $a_{j,i}$ in Fig.~\ref{fig:xor2}). Since $G_2^{i,j}$
is triconnected, this determines the rotation scheme of all its
vertices. To obtain \GammaB, compose the embeddings of all the clause
gadgets in such a way that each clause gadget lies on the outer face
of all the others.

We prove that \sefesolution is a solution of the \maxsefep instance,
namely that at most $k^*$ edges of \Gint have a different drawing in
\GammaR and in \GammaB.  Since \Gint is composed of $3$-cycles, this
corresponds to saying that at most $k^*$ of such $3$-cycles have a
different embedding in \GammaR and in \GammaB (where the embedding of
a $3$-cycle is defined by the clockwise order of the vertices on its
boundary).  In fact, a $3$-cycle with a different embedding in \GammaR
and in \GammaB can always be realized by drawing only one of its edges
with a different curve in the two drawings.
By this observation and by the fact that at most $k = k^*$ clauses are
not satisifed by $A$, the following claim is sufficient to prove the
statement.

\begin{cl}\label{cl:satisfied-same-embedding}
  For each clause $(l_i \oplus l_j) \in F$, if $(l_i \oplus l_j)$ is
  satisifed by $A$, then both $V_{i,j}$ and $V_{j,i}$ have the same
  embedding in \GammaR and in \GammaB, while if $(l_i \oplus l_j)$ is
  not satisifed by $A$, then exactly one of them has the same
  embedding in \GammaR and in \GammaB.
\end{cl}
\remove{
  \begin{proofsketch}
    Consider a clause $(l_i \oplus l_j) \in F$, where $l_i \in
    \{x_i,\overline{x_i}\}$ and $l_j \in \{x_j,\overline{x_j}\}$.
    First note that $V_{i,j}$ has the same embedding in \GammaR and in
    \GammaB, independently of whether $(l_i \oplus l_j)$ is satisfied
    or not, by construction of \GammaB and by the fact that the flip
    of $V_{i,j}$ in \GammaR is the same in the two cases in which
    $l_i$ evaluates to \texttt{true} in $A$, that are depicted in
    Fig.~\ref{fig:flips}.
    Hence, it remains to prove that, if $(l_i \oplus l_j)$ is
    satisifed by $A$, then also $V_{j,i}$ has the same embedding in
    \GammaR and in \GammaB.  First, with a case analysis analogous to
    the one depicted in Fig.~\ref{fig:flips}, one can observe that the
    flip of $V_{i,j}$ and of $V_{j,i}$ in \GammaR only depend on the
    evaluation of $l_i$ and $l_j$, respectively, in $A$. Hence, since
    one of $l_i$ and $l_j$ evaluates to \texttt{true} in $A$ and the
    other one to \texttt{false}, the flip of $V_{i,j}$ and of
    $V_{j,i}$ in \GammaR are ``opposite'' to each other. Further, by
    the construction of the triconnected clause gadget $G_2^{i,j}$,
    $3$-cycles $V_{i,j}$ and $V_{j,i}$ have ``opposite'' flips also in
    \GammaB. Since $V_{i,j}$ has the same embedding in \GammaR and in
    \GammaB, the statement of the claim follows.
  \end{proofsketch}
}
\begin{proof}
  Consider a clause $(l_i \oplus l_j) \in F$, where $l_i \in
  \{x_i,\overline{x_i}\}$ and $l_j \in \{x_j,\overline{x_j}\}$. First,
  we prove that $V_{i,j}$ has the same embedding in \GammaR and in
  \GammaB, independently of whether $(l_i \oplus l_j)$ is satisfied or
  not. Namely, the flip of $G_1^i$ selected in the construction of
  \GammaR is such that the rotation scheme of $a_{i,j}$ in \GammaR is
  $(a_{i,j}, b_{i,j})$, $(a_{i,j}, c_{i,j})$, $(a_{i,j}, c_x)$ if and
  only if $l_i$ evaluates to \texttt{true} in $A$ (where $c_x=c_i$ if
  $(l_i \oplus l_j)$ is the first considered clause involving either
  $x_i$ or $\overline{x_i}$ in the construction of \Gr, otherwise
  $c_x=c_{i,h}$ where $(l_i \oplus l_h)$ (or $(l_h \oplus l_i)$) is
  the clause involving either $x_i$ or $\overline{x_i}$ considered
  before $(l_i \oplus l_j)$ in the construction of \Gr). This can be
  easily verified by considering the flip of $G_1^i$ in \GammaR in the
  two cases in which $l_i$ evaluates to \texttt{true} in $A$, namely
  when either $x_i=$ \texttt{true} and $l_i=x_i$ or when $x_i=$
  \texttt{false} and $l_i=\overline{x_i}$, that are depicted in
  Fig.~\ref{fig:flips}. Recall that, by construction, the rotation
  scheme of $a_{i,j}$ in \GammaB is $(a_{i,j}, b_{i,j})$, $(a_{i,j},
  c_{i,j})$, and $(a_{i,j}, a_{j,i})$ if and only if $l_i$ evaluates
  to \texttt{true} in $A$.  Since $c_x$ lies outside $V_{i,j}$ in
  \GammaR and $a_{j,i}$ lies outside $V_{i,j}$ in \GammaB, the
  embedding of $V_{i,j}$ is determined by the evaluation of $l_i$ in
  $A$ in the same way in \GammaR as in \GammaB.

  Hence, it remains to prove that, if $(l_i \oplus l_j)$ is satisifed
  by $A$, then also $V_{j,i}$ has the same embedding in \GammaR and in
  \GammaB. Suppose that $l_j$ evaluates to \texttt{false} in $A$. By
  construction, the flip of $G_1^j$ selected in the construction of
  \GammaR is such that the rotation scheme of $a_{j,i}$ in \GammaR is
  $(a_{j,i}, c_{j,i})$, $(a_{j,i}, b_{j,i})$, $(a_{j,i}, c_x)$ (where
  $c_x$ is defined as above). This can be easily verified by
  considering the flip of $G_1^i$ in \GammaR in the two cases in which
  $l_j$ evaluates to \texttt{false} in $A$, namely when either $x_j=$
  \texttt{false} and $l_j=x_j$ or when $x_j=$ \texttt{true} and
  $l_j=\overline{x_j}$. Further, since $(l_i \oplus l_j)$ is satisifed
  by $A$ and $l_j$ evaluates to \texttt{false}, $l_i$ evaluates to
  \texttt{true}. Hence, by construction, the rotation scheme of
  $a_{i,j}$ in \GammaB is $(a_{i,j}, b_{i,j})$, $(a_{i,j}, c_{i,j})$,
  $(a_{i,j}, a_{j,i})$. Since $G_2^{i,j}$ is triconnected, the
  rotation scheme of $a_{j,i}$ in \GammaB is $(a_{j,i}, c_{j,i})$,
  $(a_{j,i}, b_{j,i})$, $(a_{j,i}, a_{i,j})$. Since $c_x$ lies outside
  $V_{j,i}$ in \GammaR and $a_{i,j}$ lies outside $V_{j,i}$ in
  \GammaB, the embedding of $V_{j,i}$ is the same in \GammaR and in
  \GammaB when $l_j$ evaluates to \texttt{false} in $A$.

  The fact that the embedding of $V_{j,i}$ be the same in \GammaR and
  in \GammaB when $l_j$ evaluates to \texttt{true} in $A$ (and hence
  $l_i$ evaluates to \texttt{false}) can be proved analogously.
\end{proof}


Suppose that \maxsefeinstance{k^*} admits a solution
\sefesolution. Assume that \sefesolution is optimal, that is, there
exists no solution of \maxsefeinstance{k^*} with fewer edges of \Gint
drawn differently. We construct a truth assignment $A$ that is a
solution of $\langle B,F,k \rangle$, as follows.
For each variable $x_i$, with $i=1,\dots,l$, assign \texttt{true} to
$x_i$ if the rotation scheme of $a_i$ in \GammaR is $(a_i,v_i)$,
$(a_i,b_i)$, $(a_i,d_i)$. Otherwise, assign \texttt{false} to $x_i$.

We prove that $A$ is a solution of the \xorsatp instance, namely that
at most $k$ clauses of $B$ are not satisfied by $A$.
Since \sefesolution is optimal, for any $3$-cycle $V_{i,j}$ of \Gint,
at most one edge has a different drawing in \GammaR and in
\GammaB. Also, for any clause $(l_i \oplus l_j)$, at most one of
$V_{i,j}$ and $V_{j,i}$ has an edge drawn differently in \GammaR and
in \GammaB, as otherwise one could flip \remove{clause gadget}
$G_2^{i,j}$ in \GammaB (that is, revert the rotation scheme of all its
vertices) and draw all the edges of $V_{i,j}$ and $V_{j,i}$ with the
same curves as in \GammaR.
Since $k=k^*$, the following claim is sufficient to prove the
statement.

\begin{cl}\label{cl:same-drawing-implies-satisfied}
  For each clause gadget $G_2^{i,j}$ such that $V_{i,j}$ and $V_{j,i}$
  have the same drawing in \GammaR and in \GammaB, the corresponding
  clause $(l_i \oplus l_j)$ is satisfied by $A$.
\end{cl}
\begin{proof}
  Consider a clause gadget $G_2^{i,j}$ and the drawing of the
  corresponding clause-variable gadgets $V_{i,j}$ and $V_{j,i}$ in
  \GammaB. Note that, since $G_2^{i,j}$ is triconnected, if the
  rotation scheme of $a_{i,j}$ is $(a_{i,j},b_{i,j})$,
  $(a_{i,j},c_{i,j})$, $(a_{i,j},a_{j,i})$, then the rotation scheme
  of $a_{j,i}$ is $(a_{j,i},c_{j,i})$, $(a_{j,i},b_{j,i})$,
  $(a_{j,i},a_{i,j})$. Otherwise, both the rotation schemes are
  reversed.
  Also, consider the clause-variable gadget $V_{i,j}$ corresponding to
  any clause $(l_i \oplus l_j)$ or $(l_j \oplus l_i)$ involving a
  variable $x_i$. Note that, if the rotation scheme of $a_{i,j}$ in
  \GammaR is $(a_{i,j},b_{i,j})$, $(a_{i,j},c_{i,j})$, $(a_{i,j},c_x)$
  (where $c_x$ is defined as in the proof of
  Claim~\ref{cl:satisfied-same-embedding}), then either edge
  $(b_{i,j},b_i)$ exists in \Gr and the rotation scheme of $a_{i}$ is
  $(a_{i},v_{i})$, $(a_{i},b_{i})$, $(a_{i},d_{i})$, or edge
  $(b_{i,j},d_i)$ exists in \Gr and the rotation scheme of $a_{i}$ is
  $(a_{i},v_{i})$, $(a_{i},d_{i})$, $(a_{i},b_{i})$. In both cases,
  literal $l_i$ evaluates to \texttt{true} in $A$. In fact, in the
  former case $l_i=x_i$ and $x_i$ is \texttt{true} in $A$, while in
  the latter case $l_i=\overline{x_i}$ and $x_i$ is \texttt{false} in
  $A$, by the construction of \maxsefeinstance{k^*} and by the
  assignment chosen for $A$. Analogously, if the rotation scheme of
  $a_{i,j}$ is the opposite, then $l_i$ evaluates to \texttt{false} in
  $A$.

  Consider any clause gadget $G_2^{i,j}$ such that $V_{i,j}$ and
  $V_{j,i}$ have the same drawing in \GammaR and in \GammaB. By
  combining the observations on the relationships among the rotation
  schemes of the vertices belonging to the clause gadget $G_2^{i,j}$,
  to the clause-variable gadgets $V_{i,j}$ and $V_{j,i}$, and to the
  variable gadgets $V_i$ and $V_j$, it is possible to conclude that
  $l_i$ evaluates to \texttt{true} in $A$ if and only if $l_j$
  evaluates to \texttt{false} in $A$, that is, $(l_i \oplus l_j)$ is
  satisfied by $A$. \small
\end{proof}

This concludes the proof of the theorem.
\end{proof}

\section{Conclusions}\label{se:conclusions}

In this paper we proved several results concerning the computational
complexity of some problems related to the \sefep and the \ptckpbep
problems. We showed that the version of \sefe in which all graphs
share the same intersection graph \Gint (\sunsefep) is \NPC for $k\geq
3$ even when \Gint is a tree and all the input graphs are
biconnected. This improves on the result by
Schaefer~\cite{s-ttphtpv-13} who proved \NPCN when \Gint is a forest
of stars and two of the input graphs consist of disjoint biconnected
components. Further, we prove \NPCN of problem \ptckpbepshort for
$k\geq 3$ when $T$ is a caterpillar and two of the input graphs are
biconnected, and of problem \pkpbepshort for $k\geq 3$. These results
improve on the previously known \NPCN for $k$ unbounded by
Hoske~\cite{hoske-befpa-12}. Also, we provided a linear-time algorithm
to decide \ptckpbepshort for $k\geq 2$ when $k-1$ of the input graphs
are $T$-biconnected. Most notably, this result enlarges the set of
instances of \ptcTWOpbepshort, and hence of the long-standing open
problem \sefep when \Gint is connected, for which a polynomially-time
algorithm is known. For this problem, we also proved that all
the instances can be encoded by equivalent instances in which one of
the two graphs is biconnected and series-parallel.  It is also known
that the biconnectivity of both the input graphs suffices to make the
problem polynomial-time solvable~\cite{br-drpse-13}.  On one hand, our
results push \ptcTWOpbepshort closer to the boundary of polynomiality.
On the other hand, since we proved that for $k\geq 3$ the
biconnectivity of all the input graphs does not avoid \NPCN, it is
natural to wonder whether dropping the biconnectivy condition on one
of the two graphs in the case $k=2$ would make it possible to simulate
the degrees of freedom that are given by the fact of having more
graphs.

Moreover, we considered the optimization version \maxsefep of \sefe
with $k=2$, in which one wants to draw as many common edges as
possible with the same curve in the drawings of the two input
graphs. We showed \NPCN of this problem even under strong restrictions
on the embedding of the input graphs and on the degree of the
intersection graph that are sufficient to obtain polynomial-time
algorithms for the original decision version of the problem.

\section*{Acknowledgements}
A preliminary version of this paper appeared at the $8^{th}$ International Workshop on Algorithms and Computation (WALCOM'14)~\cite{adn-osnsp-14}.

Part of the research was conducted in the framework of ESF
  project 10-EuroGIGA-OP-003 GraDR "Graph Drawings and
  Representations", of 'EU FP7 STREP Project "Leone: From Global
  Measurements to Local Management", grant no. 317647', and of the MIUR project AMANDA ``Algorithmics for MAssive and Networked DAta'', prot. 2012C4E3KT\_001.

The authors would like to express their gratitude to Christopher Auer, Andreas Ga{\ss}ner,
and Ignaz Rutter for useful discussions about the topics of this
paper.

{\bibliography{bibliography}} 

\end{document}